\documentclass[runningheads]{llncs}
\usepackage[T1]{fontenc}
\usepackage{amsfonts}
\usepackage{parskip}
\usepackage{enumerate}
\usepackage{amsmath}
\usepackage{latexsym}
\usepackage{amssymb}
\usepackage{mathrsfs}
\usepackage{cases}
\usepackage{subcaption} 
\usepackage{subfloat}
\usepackage{graphicx}
\usepackage[hidelinks,bookmarksopen=true]{hyperref}
\usepackage{bookmark}

\usepackage{floatrow}

\begin{document}
%
%\title{Contribution Title\thanks{Supported by organization x.}}
\title{Ring Signature from Bonsai Tree: How to Preserve the Long-Term Anonymity}
%\author{Mingxing Hu \and Zhen Liu}

\author{Mingxing Hu \and Yunhong Zhou}

% \institute{Shanghai Jiao Tong University, Shanghai, China\\
% \email{\{mxhu2018,liuzhen\}@sjtu.edu.cn}}%

 \institute{Shanghai Jiao Tong University, Shanghai, China\\
 \email{\{mxhu2018,zhouyunhong\}@sjtu.edu.cn}}%

\maketitle              % typeset the header of the contribution
\begin{abstract}
Signer-anonymity is the central feature of ring signatures, which enable a user to sign messages on behalf of an arbitrary set of users, called the ring, without revealing exactly which member of the ring actually generated the signature. Strong and long-term signer-anonymity is a reassuring guarantee for users who are hesitant to leak a secret, especially if the consequences of identification are dire in certain scenarios such as whistleblowing. The notion of \textit{unconditional anonymity}, which protects signer-anonymity even against an infinitely powerful adversary, is considered for ring signatures that aim to achieve long-term signer-anonymity. However, the existing lattice-based works that consider the unconditional anonymity notion did not strictly capture the security requirements imposed in practice, this leads to a realistic attack on signer-anonymity.

In this paper, we present a realistic attack on the unconditional anonymity of ring signatures, and formalize the unconditional anonymity model to strictly capture it. We then propose a lattice-based ring signature construction with unconditional anonymity by leveraging bonsai tree mechanism. Finally, we prove the security in the standard model and demonstrate the unconditional anonymity through both theoretical proof and practical experiments.

\keywords{Ring signatures  \and Unconditional anonymity \and Bonsai tree  \and Lattices.}
\end{abstract}
\section{Introduction}\label{Introduction}\pdfbookmark[1]{Introduction}{Introduction}
\textit{Ring signatures}, originally introduced by Rivest et al.~\cite{RST01}, enable a signer to hide in a \textit{ring} of potential signers, without requiring any central coordination, as the rings can be formed in a spontaneous manner. Ring signatures have many natural applications, such as the ability to leak secrets while staying anonymous within a certain set, i.e., whistleblowing~\cite{RST01}, and recently certain types could be used as a building block for cryptocurrencies~\cite{EZS+19,TKS+19,TSS+18}. Ring signatures have been extensively studied in various flavors, including RSA~\cite{DKN+04}, symmetric-key~\cite{DRS18,KKW18}, bilinear pairings~\cite{BDR15,CWL+06,MS17}, (non)interactive proof systems~\cite{BDH+19,BKM06}, and lattices~\cite{BLO18,CGH+21,EZS+19,LAZ19,MBB+13,PS19}. 

%e to the features of flexibility and signer-anonymity, ring signatures has many natural applications,

\textbf{The infinitely powerful adversary's ability was underestimated.} 
The notion of unconditional anonymity was initially presented by Rivest et al.~\cite{RST01}, which claimed that the signer-anonymity is preserved even when confronted with an \textit{infinitely powerful adversary}, i.e., an adversary with unlimited computational resources and time. Although this work lacks a formal model to capture the unconditional anonymity notion, an exact description was given: ``\textit{Even an infinitely powerful adversary with access to an unbounded number of chosen-message signatures produced by the same ring member still cannot guess his identity}''. There is one word in this description that needs to be highlighted, that is ``\emph{unbounded} number of chosen-message signatures produced by the \emph{same} ring member''. However, it was not given serious consideration in the subsequent works~\cite{LAZ19,TKS+19,TSS+18,TSS+20,WZZ18,ZHX+16}.

On the anonymity model, these works did not sufficiently capture the adversary's ability, i.e., allow the adversary to obtain an arbitrary number of chosen-message signatures produced by the \emph{same} ring member. On the construction, these works employed rejection sampling method \cite{Lyu09,Lyu12} or Gaussian sampling method \cite{GPV08,MP12} to generate signatures. These sampling algorithms have the property of producing signatures that are statistically close to uniform, meaning that even an unbounded adversary cannot distinguishing the real signer's identity from the signatures distribution, as the distinguishing advantage is only negligible. Based on this fact, these works claimed the unconditional anonymity is hold in their proofs. However, the point that the infinitely powerful adversary can exploit is the negligible distinguishing advantage on these signature tuples, which is \emph{all-but-one} (cf. Section \ref{OurMethods} for details). The all-but-one negligible distinguishing advantage on the signatures can accumulate to be non-negligible when the adversary collected a sufficient amounts of signature instances with respect to the same ring member. We confirmed this observation through a theoretical attack strategy and experimental results as shown in Section \ref{OurMethods}.

%We later show this neglect can lead to fatal damage i.e., breaking the signer-anonymity. 

%There is always a negligible distinguishing advantage that may be derived from the distribution of the signature, public/signing keys, messages, or randomness (including signing randomness), etc., which can help to identify the identity of the real signer since even the negligible advantage is can be exploited by the infinitely powerful adversary. In summary, the ability of an infinitely powerful adversary was underestimated in prior works.

%There are three keywords in this description that need to be highlighted. The first is ``infinitely powerful adversary'' rather than ``limited powerful adversary''. But in many subsequent works~\cite{BKM06,LAZ19,WZZ18,ZHX+16}, the unconditional anonymity notion is turned into a \emph{compatible} signer-anonymity notion and the infinitely powerful adversary obeys the rules defined by other signer-anonymity notions. For instance, in work~\cite{BKM06}, the infinitely powerful adversary in the setting of \textit{unconditionally signer-anonymity against adversarially-chosen-key attack} is not allowed to obtain the secret/signing keys of the honest ring members as \textit{unconditionally signer-anonymity against full key revealed}. This limitation is \emph{inappropriate} since the infinitely powerful adversary has unlimited computational time and resources who can even compute the secret keys of the cryptosystems for which based on computational hardness assumptions. 

%The third keyword is ``any advantage'' rather than ``negligible advantage''. 

\textbf{Relying on random oracle heuristics isn't reassuring.} 
Almost all the ring signatures with unconditional anonymity prefer to rely on the ideal random distribution of a cryptographic hash function that is modeled as a random oracle. 
However, there is also a negligible distinguishing advantage with the uniform distribution when the random oracle is instantiated by a hash function such as SHA-3~\cite{NISTSHA3} in practice, and hence, it also suffers the threat from the infinitely powerful adversary. 

Prior works~\cite{CGH04,DOP05} have pointed out that the proof in the random oracle model could lead to insecure schemes when the random oracle is implemented in practical scenarios. Quantum random oracle~\cite{BDF+11} can be seen as a stronger notion of random oracle, but as shown in~\cite{ES20}, the security holds in the quantum random oracle does not imply security in the standard model. Chatterjee et al. \cite{CGH+21} formalized the security models for ring signatures in quantum setting, attempting to capture adversaries with quantum access to the signer, but as pointed by \cite{CCL+22}, it behaves differently from ordinary signatures, and it is unclear if their models are as strong as the standard security notion when restricted to the classical world. And recently, Branco et al. \cite{BDW22} presented a novel ring signature in the standard model, explained why their work cannot rely on random oracle, and introduced the ramifications of relying on random oracles in their construction. Therefore, proving unconditional anonymity in the standard model is more desirable and reassuring for preserving \emph{long-term} signer-anonymity.

\subsection{Our Contributions}\pdfbookmark[2]{Our Contributions}{Our Contributions}

Our contributions are mainly three-folds: A realistic attack on the ring signatures' unconditional anonymity, the unconditional anonymity model for ring signatures, and a new signer-unconditional-anonymous ring signature scheme.

\textbf{A simple yet effective realistic attack on the unconditional anonymity of ring signatures.} The strategy under our realistic attack is simple, that is accumulating the negligible distinguishing advantage to be non-negligible by employing a basic fact of ring signatures and the classic measurement tool -- \emph{statistical distance}. Furthermore, our attack algorithm is effective, we first demonstrate it by presenting a theoretical attack strategy, then the strategy is firmly affirmed by conducting a double experiments. Specifically, due to the statistical distance is defined on the concrete probability for each element of instance, so we first evaluate the average probability that our attack procedure successfully identifying the real signer for a ring signature scheme with a minimum ring size. Based on that, we conduct the second experiment to show the distinguishing advantage on the signature distributions of prior works can be accumulated to a non-negligible level. 

%Using the similar method, we also conduct for our work.

%The fact is basic, that is the statistical distance between different instances can be accumulated on the same position.

\textbf{Formalize the unconditional anonymity model.} 
We formalize the unconditional anonymity model for ring signatures, which sufficiently considers the ability of an infinitely powerful adversary, rather than simply change the adjective of the adversary from `PPT' to `any' as previous works. Specifically, we redefine the challenge phase so that the adversary is allowed to obtain an arbitrary number of signatures with respect to the same ring member. In this setting, our model captures the realistic attack presented above strictly since it provides a more comprehensive understanding of unconditional anonymity for ring signatures.

\textbf{A new signer-unconditional-anonymous ring signature scheme: New construction method, Standard model, Strong security notions, and Practical feasibility.}

    \begin{itemize}
        \item We introduce the bonsai tree mechanism \cite{CHKP10} to construct unconditional signer-anonymous ring signature scheme. By using insights on the bonsai tree, we are able to eliminate the all-but-one negligible distinguishing advantage for the signature distributions. As a result, the realistic attack mentioned earlier is immune as the attack result cannot be considered as evidence to identify the real signer. We elaborate it in Section \ref{OurMethods}.

        \item We prove the security (unforgeability and unconditional signer-anonymity) in the standard model, i.e., without resorting to any (quantum) random oracle. Due to the ideal distribution of random oracle heuristics cannot be guaranteed in practice, and given the significant power of an unbounded adversary, realizing the unconditional anonymity in the standard model is more reassuring for preserving \emph{long-term} signer-anonymity.

        \item The security notions of our work is strong. On the unforgeability,  the adversary is allowed to corrupt honest ring members within a ring and obtain their signing keys and even the randomness used to generate these keys. This is referred to as unforgeability w.r.t. insider corruption, which is the strongest one among the unforgeability notions presented by Bender et al. \cite{BKM06}. As for anonymity, the signer-anonymity w.r.t. full key exposure (it is also the strongest one among the anonymity notions presented by Bender et al. \cite{BKM06}) is implied by our model.

        \item Our ring signature scheme is tightly secure and has an asymptotically efficient signature size. Specifically, by leveraging the key homomorphic evaluation algorithm, the signature size is effectively reduced in our construction. And in the unforgeability proof, the Gaussian parameter $\sigma_{\textsc{Saml}}$, which dominates the concrete signature size in practice, is sublinear with the ring size rather than linear as the scheme theoretically shows. Additionally, we instantiate our work from compact lattices over rings, and we also provide the implementation results to prove its practical feasibility.

        %Furthermore, our work achieves the strong security guarantees without compensating the compactness.
    \end{itemize}

\subsection{Related Work}\label{RelatedWork}\pdfbookmark[2]{Related Work}{Related Work}
We give further details with respect to the related works. Since the ring signatures have been presented, it has been studied extensively from various assumptions. The initial work \cite{RST01}, presented an elegant construction with the only assumption of classical RSA. This work is unconditional signer-anonymity but their proofs only hold in the random oracle model and without formal models.

The ring signature schemes in standard model are proposed concurrently by Chow et al. \cite{CWL+06} and Bender et al. \cite{BKM06}. Specifically, Chow et al. \cite{CWL+06} proposed a ring signature scheme in the standard model from pairings, namely, is not post-quantum secure. As aforementioned, Bender et al. \cite{BKM06} formalized a hierarchy of security models to model the realistic attacks in practice for ring signatures, based on that this work presented a generic ring signature scheme from ZAP proof system in the standard model, and the construction is asymptotically-efficient. Backes et al. \cite{BDH+19} presented an improved work from \cite{BKM06}, which preserved the merits of asymptotically-efficient and in standard model, but did not consider the unconditional anonymity. Bose et al. \cite{BDR15} presented a constant size ring signature scheme in the standard model, i.e., the signature size is independent with the ring size, but the based assumption is not quantum resistant. Malavolta and Schr{\"{o}}der \cite{MS17} proposed an efficient ring signature scheme from bilinear groups and knowledge of exponent assumption, which is the first almost practical ring signature in the standard model, but is not post-quantum secure.

The research on lattice-based ring signatures has drawn more attention from the community since cryptographic primitives when instantiated from lattice assumptions enjoy distinctive merits such as post-quantum secure. The work of Lu et al. \cite{LAZ19} presented a generic ring signature construction based on the framework of Rivest et al. \cite{RST01}, which is almost as efficient
as discrete-log or pairing-based counterparts, but the random oracle is required. Esgin et al. \cite{EZS+19} presented a scalable ring signature, as the building block for their cryptocurrency protocol, which has the shortest signature size at the time, but relied on random oracle. Chatterjee et al. \cite{CGH+21} proposed a compact ring signature scheme and without random oracle, which is taken as the state-of-the-art work, but it did not consider the unconditional anonymity. The work of Melchor et al. \cite{MBB+13} is adapted from a standard signature scheme, which is asymptotically-efficient, and its anonymity model is strong even compare with the classic full key exposure model \cite{BKM06}, but need to rely on random oracle heuristics. Recently, Park and Sealfon \cite{PS19} proposed ring signatures with novel security notions, among that, the ring signature with the security notion of unclaimability is instantiated from standard lattice assumption and in the standard model. However, this work did not consider the unconditional anonymity and not asymptotically-efficient since its signature size grows quadratically with the ring size.

There are also ring signatures which are post-quantum secure but not lattice-based. Katz et al. \cite{KKW18} instantiate the ``MPC-in-the-head'' paradigm by employing MPC methods, which induces a post-quantum secure scheme with asymptotically-faster and concretely efficient, but needs to rely on the random oracle. Derler et al. \cite{DRS18} proposed post-quantum secure accumulators from symmetric-key primitives, based on that, this work constructed an asymptotically-efficient ring signature scheme, but its signer-anonymity is not unconditional and also rely on the random oracle.

\section{Our Methods}\label{OurMethods}\pdfbookmark[1]{Our Methods}{OurMethods}

In this section, we will begin by demonstrating our attack strategy on the unconditional anonymity of existing works, as mentioned earlier. Then we describe how to employ bonsai tree to construct a ring signature scheme that provides unconditional anonymity.

\textbf{Attack strategy.} To describe the attack strategy, it is instructive to abstract a \emph{basic} construction from prior works~\cite{LAZ19,TKS+19,TSS+18,TSS+20,WZZ18,ZHX+16}. Let $\textsf{R}=\{\textsf{vk}^{(1)},\dots,\textsf{vk}^{(N)}\}$ be a ring with size $\lvert\textsf{R}\rvert=N$, each ring member has a pair of verification/signing key $(\textsf{vk}^{(i)} = a^{(i)}, \textsf{sk}^{(i)} = 
t^{(i)})$ for $i\in[N]$. Let $a^{(i)}(\cdot)$ be the hash function, that belongs to some hash family, defined by $a^{(i)}$. Let $H$ be a function for hashing the message $\mu$. The signing procedure is one of the $N$ members who samples ${e}^{(1)},{e}^{(2)},\dots,{e}^{(N)}\leftarrow  D$ from a specific distribution $ D$ such that the following \emph{ring equation} holds
\begin{equation}\label{equ:eq1}
a^{(1)}({e}^{(1)}) +a^{(2)}({e}^{(2)}) +\dots+ a^{(N)}({e}^{(N)}) = H(\mu)
\end{equation}More specifically, the ring equation (\ref{equ:eq1}) is generated in two steps. Assume the index of the signer is $s$ s.t. $s \in[N]$, the signer first selects preimage ${e}^{(i)}\xleftarrow{_\$} D$ randomly for ring members with index $i\in\{N\}\setminus s$, i.e., except the signer himself, then the signer using his signing key $t^{(s)}$ to sample the ${e}^{(s)}$ by specific sampling algorithm such that the ring equation (\ref{equ:eq1}) holds and ${e}^{(s)}$ is distributed statistically close to the distribution $D$. Finally, output $\mathbf{e}:=(e^{(1)},e^{(2)},\cdots,e^{(N)})$ as the signature. The verification procedure checks if the signature $\mathbf{e}$ is well-formed and if the ring equation (\ref{equ:eq1}) holds, accept, otherwise reject the signature.

The signer-anonymity property is preserved by the property of the based sampling algorithm such as the Gaussian sampling algorithm employed in works \cite{LAZ19,WZZ18,ZHX+16} or the rejection sampling algorithm used in works \cite{TKS+19,TSS+18,TSS+20}, which makes the distribution on the $e^{(s)}$ statistically close to $D$. It means the advantage to distinguish that is negligible even for an infinitely powerful adversary. However, the negligible distinguishing advantage is not zero, and the advantage is an \emph{all-but-one} negligible distinguishing advantage\footnote{The all-but-one negligible distinguishing advantage is that all the elements of an instance is randomly selected from uniform distribution (whose distinguishing advantage is zero) except one element is selected by specific algorithms (whose distinguishing advantage is negligible rather than zero). For example, the signature instances given in Figure \ref{fig1} are with all-but-one negligible distinguishing advantage since the elements in the dashed box are all sampled by specific algorithms while the remained are sampled randomly.}, which is still can be exploited by the infinitely powerful adversary. Recall the general definition of ring signatures, the ring is an \emph{ordered set}, it means the signer's index is corresponded to the signature component $e^{(i)}$ of the signature tuple $\mathbf{e}$. 

In this setting, the all-but-one negligible distinguishing advantage all locate on the same position, i.e., the position on index $s$ (cf. the dashed box depicted in Figure \ref{fig1}). Therefore, the adversary can \emph{accumulate} the negligible advantage to be non-negligible when obtained sufficient amounts of instances from the \emph{same} ring member (assume his index in the ring is $s$). As we know, for sufficiently large $n$ and every constant $c>0$, a function $f$ in $n$ is negligible if it vanishes faster than the inverse of any polynomial in $n$ i.e., $f(n)<n^{-c}$, and a function $f$ in $n$ is super-polynomial if $f(n)>n^{c}$. Super-polynomial and negligible functions are reciprocals of each other. Therefore, theoretically if the adversary obtains sufficient amounts of message-signature tuples w.r.t. a \emph{fix} ring member, then the advantage to distinguish the distribution on $e^{(s)}$ from $ D$ can accumulate to be non-negligible i.e., $\geq n^{-c}$.

%This method is also worked for the negligible distinguishing advantage caused by other events such as the instantiating of the random oracle function, or adversarially-chosen ring or message queries, etc. 

%($O(2^{nc})$)

\begin{figure}
\includegraphics[width=0.8\textwidth]{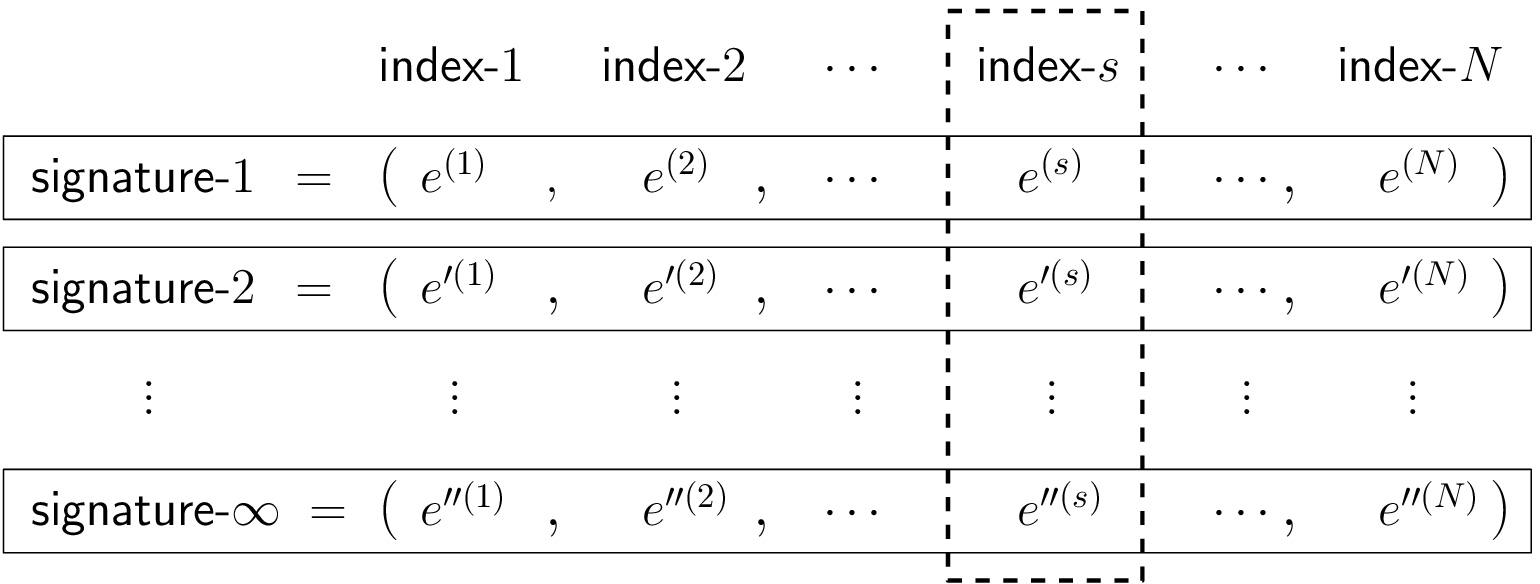}
\caption{Distinguishing advantage accumulation. An infinitely powerful adversary is able to accumulate the negligible distinguishing advantage to be non-negligible when the adversary obtains sufficient amounts of message-signature tuples w.r.t. a \emph{fix} ring member (e.g., the ring member with index $s$ in the dashed box). The symbol $\infty$ represents an uncertain numerical value whose magnitude depends on the number of samples required by the adversary to accumulate a non-negligible advantage.} 
\label{fig1}
\end{figure}

As the prior works~\cite{LAZ19,TKS+19,TSS+18,TSS+20,WZZ18,ZHX+16}, we also use the \emph{statistical distance} to measure how different are two probability distributions. The algorithmic form of our attack strategy is the computation of the statistical distance with respect to two distributions $X$ and $Y$, that is \begin{equation}\Delta(X,Y)=\frac{1}{2}\sum_{x\in S}\big\lvert \Pr[X=x] - \Pr[Y=x]\big\rvert\end{equation}where we assume $S$ as a finite set that $X,Y$ taking values. Below we show that when the size of the set $S$ exceeds certain amount, the non-negligible distinguishing advantage can be  accumulated from the evaluation on statistical distance.

%用二分法优化攻击，从而快速用公共交集找出签名者的身份。一个fact是ring的最小是2个环成员。

% 一条线表示这些实例都是真实生成的，而且都是由一个签名者生成的。(All-by-one)

%一条线表示这些实例中只有一个实例或者多项式个实例是真实生成的，其它实例都是均匀随机选取的。（攻击方式是将这个（些）实例与均匀随机的实例不断迭代地相对比）。(Only-1-or-poly)

%一条线表示这些实例全部都是均匀随机选取的。(All-Uniform)

%一条线表示全部实例中的全部元素都是由trapdoor提取的。(All-by-Trapdoor)

%左右两个图，左图先来个toy的ring signature有5个用户那种，然后指定一个比如3号是签名者，然后柱状图计算advantage的accumulation。分析说明一下 statistical distance 的公式源自于分布中不同变量的优势差。所有才有了右图的趋势。然后说明虽然我们的计算量不足以模拟指数级别的infinitely powerful敌手，但是根据右图趋势，我们可以乐观且审慎地预测 anonymity会被破解。
%左图我们先考虑最简单的情况，比如就2个用户。

\noindent\textbf{Implementation.} Below we examine the attack strategy in practice. Due to the precisely emulating the ability of an infinitely powerful adversary is impractical, we have to show the experiment with conservative parameters setting. Specifically, we set the $n=2^{32}$ sufficiently large so that the adversary's ability is captured sufficiently, set the constant $c=1$ and ring size $\lvert R\rvert=N=2$ so that the computational cost of the cryptographic operations involved is minimized. 
%In this setting, the negligible probability should be quantified as $2^{-32}$ (at most). 

%我们实验的原理是在收集关于固定用户指数量级签名实例的情况下，可将可忽略的区分优势积累到不可忽略。

The rationale under our experiment is to exploit the statistical distance, i.e., probability differences between two different distributions. Therefore, we first evaluate the average probability of the attack procedure correctly pointing out the real signer with the varying instance number $2^{12}$ to $2^{32}$ (cf. Figure \ref{AttackStrategy:a}). The experimental result shows that as the number of instances increases, the average probability exhibits a clear increasing trend towards $1$, which confirms the effectiveness of our attack strategy. This is further supported by our subsequent experiment (cf. Figure \ref{AttackStrategy:b}), where we compare the distinguishing advantage evaluated on two different types of signature instances, the one is from uniform while the other is from the real scheme. Under our parameters setting, the experimental result shows that the distinguishing advantage can be effectively accumulated to non-negligible.

\begin{figure}[!htb]
    \ffigbox[\textwidth]
    {
        \begin{subfloatrow}[2]%useFCwidth
        \ffigbox[\FBwidth]{
           \includegraphics[width=5.1cm,height=3.8cm]{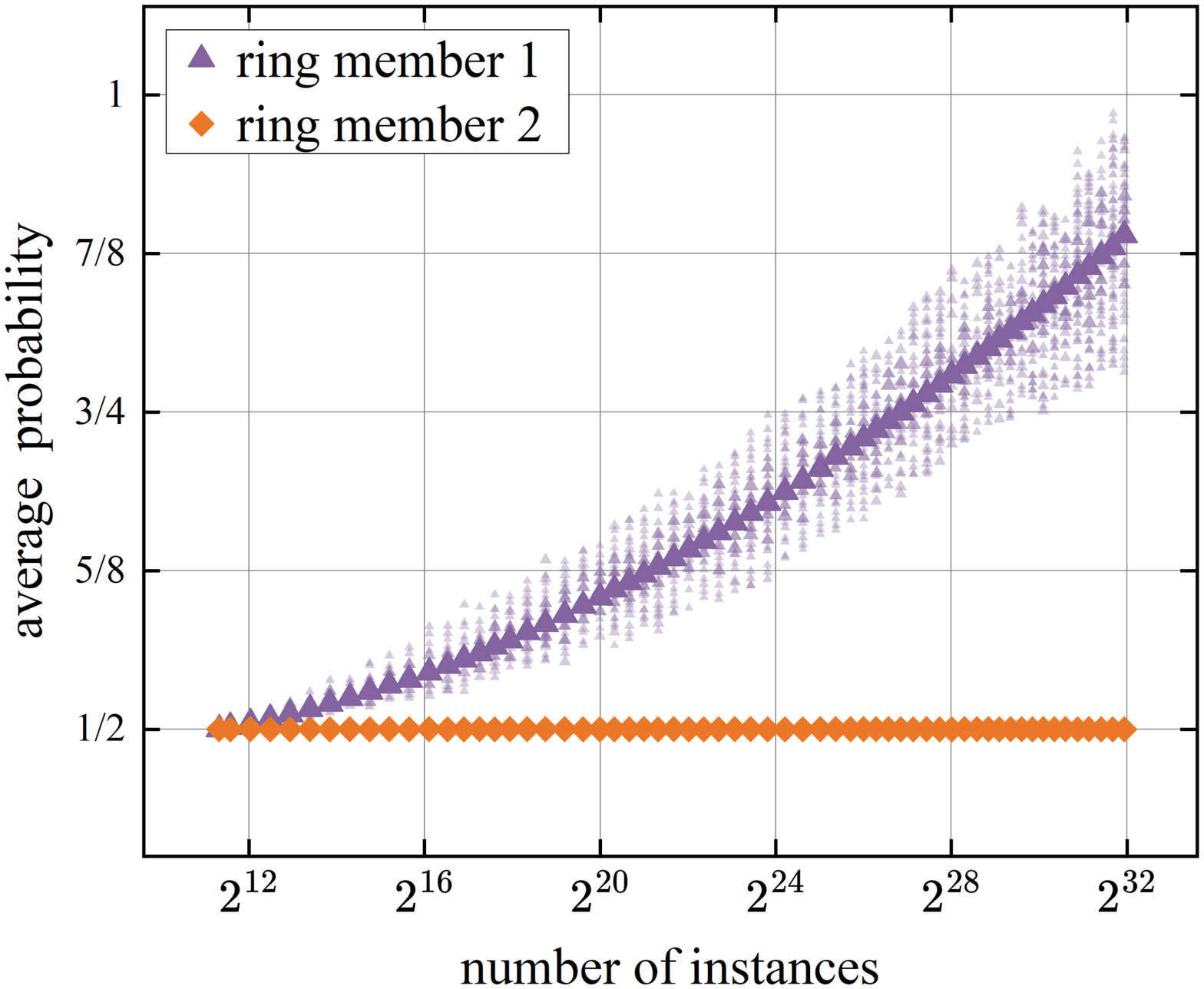}
        }{\caption{Evaluate the average probability of  correctly pointing out the real signer. The signature instances with respect to ring member-1 are truly generated while for the ring member-2 are sampled uniformly.}\label{AttackStrategy:a}}
        \ffigbox[\FBwidth]{
           \includegraphics[width=5.5cm,height=3.8cm]{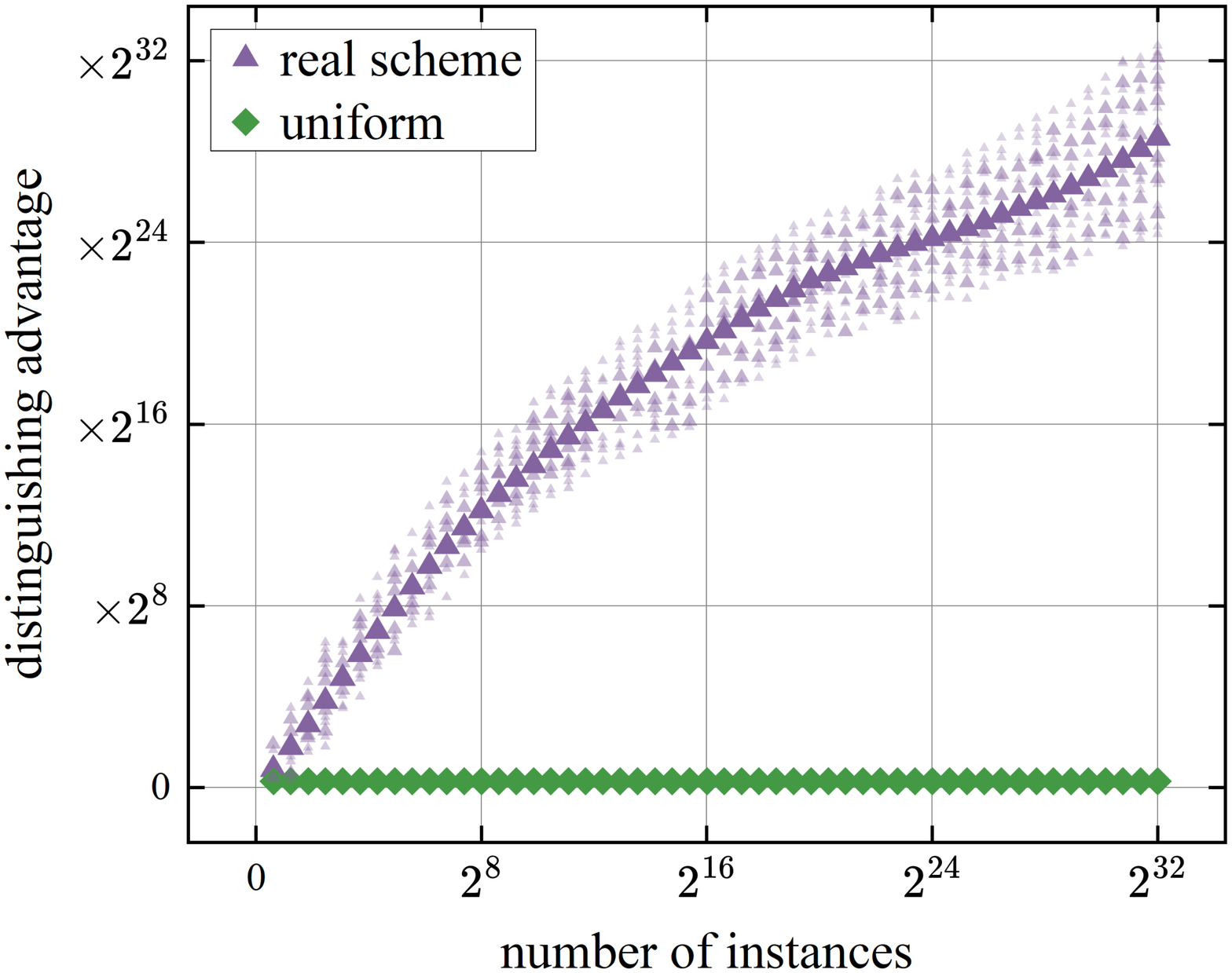}
        }{\caption{Evaluate the distinguishing advantage on two different distribution types of signature instances: from real ring signature scheme and uniform distribution. The result is based on the result of Figure \ref{AttackStrategy:a}.}\label{AttackStrategy:b}}
\end{subfloatrow}   
    }{\caption{Attack strategy implementation.}\label{AttackStrategy}}
    \end{figure}

% \begin{figure}
%     \centering
%     \subfloat[Evaluate the average probability for a ring signature with $|\textsf{R}|=2$]{\includegraphics[width=6.1cm,height=4.3cm]{toyRS.eps}}
%     \hfill 
%     \subfloat[Evaluate the distinguishing disadvantage for a full-fledged ring signature]{\includegraphics[width=5.92cm,height=4.3cm]{4LinesOfStatisticalDistance.eps}}

%     \caption{Implementation of the attack strategy.}
%     \label{fig:AttackStrategy}
% \end{figure}

% \begin{figure}[htbp]
%     \centering
%     \subfloat[IMG A]{\includegraphics[width=0.48\textwidth]{toyRSdataSource2.eps}}
%     \hfill
%     \subfloat[IMG B]{\includegraphics[width=0.48\textwidth]{toyRSdataSource2.eps}}
%     \newline
%     \subfloat[IMG C]{\includegraphics[width=0.2\textwidth]{toyRSdataSource2.eps}}
%     \hfill
%     \subfloat[IMG D]{\includegraphics[width=0.2\textwidth]{toyRSdataSource2.eps}}
%     \caption{Four Imgs}
% \end{figure}

\vspace{2mm}
\noindent\textbf{Achieving the unconditional anonymity.} We achieve the unconditional anonymity by using the \emph{bonsai tree} mechanism~\cite{CHKP10} which samples the elements of signature in \emph{one shot} rather than samples them in a sequence as in previous works. Specifically, we employ the Bonsai-Tree instantiated by Micciancio and Peikert's work \cite{MP12} and the associated trapdoor generation algorithm \textsf{TrapGen}, trapdoor delegation algorithm \textsf{TrapDel}, and the underlying Gaussian sampling algorithm \textsf{GauSample}. In this setting, the aforementioned verification/signing key pair $(a^{(i)},t^{(i)})$ is instantiated as the matrix $\mathbf {A}^{(i)}$ and its trapdoor $\mathbf T^{(i)}$ which are generated by \textsf{TrapGen}. In the signing phase, the signer (with index $s\in[N]$) concatenates all ring member's matrix $\mathbf {A}^{(i)}$ to an augmented matrix $\mathbf F=[\mathbf A^{(1)}\mid\mathbf A^{(2)}\mid\dots\mid\mathbf A^{(N)}]$, then obtains a delegated trapdoor $\mathbf T_{\mathbf F}$ of $\mathbf F$ by \textsf{TrapDel} with input the signer's trapdoor $\mathbf T^{(s)}$. Then the signer samples the signature $\mathbf{e}$ by \textsf{GauSample} with input the delegated trapdoor $\mathbf T_{\mathbf {F}}$.

We note that the process to sample the signature is in one-shot, i.e., sampling all the items that is $\mathbf{e}=[\mathbf e^{(1)}\mid\mathbf e^{(2)}\mid\dots\mid\mathbf e^{(N)}]$ in one time, instead of sampling in a sequence. In this way, each item $\mathbf e^{(i)}$ of $\mathbf{e}$ has a negligible statistical distance with the uniform distribution, rather than only the item $\mathbf e^{(s)}$ as prior works. Therefore, the negligible distinguishing advantage cannot be exploited as above since the all-but-one negligible distinguishing advantage is eliminated, namely, the accumulated non-negligible statistical distance \emph{cannot be taken as evidence} to determine the identity of the real signer. This is theoretically confirmed by our anonymity proof (cf. proof of Theorem \ref{Them:AnonymityProof}) and practically confirmed by our experimental result (cf. the implementation in Section \ref{Sec:Implementation}).

%By the property of \textsf{TrapDel}, the delegated trapdoor $\mathbf T_{\mathbf F}$ is distributed statistically independent with the original trapdoor $\mathbf T_{\mathbf {A}^{(s)}}$. And by the property of \textsf{GauSample}, the $\mathbf{e}$ is statistically independent with $\mathbf T_{\mathbf F}$ and is statistically close to the uniform distribution. This means that the distinguishing advantage remains negligible even the adversary obtained all the ring members' trapdoor.

\section{Definitions}\label{Definitions}\pdfbookmark[1]{Definitions}{Definitions}
In this section, we present the definitions of algorithms and security models for our ring signature. 

%事实上，发现了PS19它只是unclaimable，只是签名随机性的不可区分，而

%PS19所宣称的unclaimable是不成立的，因为公布了签名tuple和相应随机性后，攻击者只需验证所对应的随机性，满足SampleCond等式成立的就是真实签名者。因为其它签名者没有他的私钥。

\begin{definition}[Ring Signature]
A ring signature scheme consists of the following algorithms: \label{RS}
\begin{itemize}
 \item {\rm $\textsf{Setup}(1^n)\rightarrow \textsf{PP}$}. The system setup algorithm takes as input the security parameter $n$, and outputs the public parameters {\rm $\textsf{PP}$}.
    
    {\rm \textsf{The} $\textsf{PP}$ \textsf{are common parameters used by all ring members in the system, for example, the message space $\mathcal{M}$, the modulo, etc. To guarantee that the public has no concerns on the existing of trapdoors for PP, the randomness used in Setup can be included in $\textsf{PP}$.} \\$\textsf{In the following,}$ $\textsf{PP}$ $\textsf{is an implicit input parameter to every algorithm.}$}
    
    \item {\rm $\textsf{KeyGen}()\rightarrow (\textsf{vk, sk})$}. The key generation algorithm outputs a verification key {\rm $\textsf{vk}$} and a signing key {\rm $\textsf{sk}$}.
    
   {\rm \textsf{Any ring member can run the \textsf{KeyGen} algorithm to generate a pair of verification key and signing key.}}
    
    \item {\rm $\textsf{Sign}(\textsf{sk},\mu,\textsf{R})\rightarrow \Sigma$}. The signing algorithm takes as input a signing key {\rm \textsf{sk}}, a message $\mu\in \mathcal{M}$, and a ring of verification keys {\rm $\textsf{R}=(\textsf{vk}^{(1)},\dots,\textsf{vk}^{(N)})$}. Assume that (1) the {\rm \textsf{sk}} and the corresponding verification key {\rm \textsf{vk}} is a valid key pair output by {\rm\textsf{KeyGen}} and {\rm \textsf{vk}} $ \in $ {\rm\textsf{R}}, (2) the ring size {\rm$\lvert\textsf{R}\rvert\geq 2$}, (3) each verification key in ring {\rm$\textsf{R}$} is distinct. The algorithm outputs a signature $\Sigma$.
    
%    {\rm \textsf{Note that it is open on whether the \textsf{Sign} algorithm is probabilistic or deterministic, which may depend on the concrete constructions.}}
    
    \item {\rm $\textsf{Ver}(\mu,\textsf{R},\Sigma)\rightarrow 1/0$}. The verification algorithm takes as input a message $\mu$, a ring {\rm\textsf{R}}, and a signature $\Sigma$, the algorithm outputs 1 if the signature is valid, or 0 if the signature is invalid.

\end{itemize}
\end{definition}

\noindent$\textbf{Correctness}$. A ring signature scheme is correct if, for any $n\in \mathbb{N}$, any $N=\mathrm{poly}(n)$, any $i\in[N]$, any messages $\mu\in \mathcal{M}$, any {\rm $\textsf{PP}\leftarrow \textsf{Setup}(1^n)$} as an implicit input parameter to every algorithm, any $N$ verification/signing key pairs {\rm $(\textsf{vk}^{(1)},\textsf{sk}^{(1)}),$ $\dots,(\textsf{vk}^{(N)},\textsf{sk}^{(N)})\leftarrow \textsf{KeyGen}()$}, and any {\rm $\Sigma\leftarrow\textsf{Sign}(\textsf{sk}^{(i)},\mu,\textsf{R})$} where {\rm $\textsf{R}=(\textsf{vk}^{(1)},$ $\dots,\textsf{vk}^{(N)})$}, it holds that
\[
    {\rm \mathrm{Pr}\big[\textsf{Ver}(\mu,\textsf{R},}\Sigma{\rm )=1\big]=1-\mathrm{negl}(n)}
\]where the probability is taken over all the random coins in the experiment.

\subsection{Unforgeability Model}\label{Unforgeability Model}\pdfbookmark[2]{Unforgeability Model}{Unforgeability Model}
    
We use the unforgeability w.r.t. insider corruption model, which is the strongest among the hierarchy of unforgeability models proposed by Bender et al.~\cite{BKM06}. This model is designed to capture realistic attacks, where the adversary can adaptively corrupt honest participants within a ring and obtain their signing keys and even the randomness used to generate these keys.

\noindent$\textbf{Unforgeability}$. A ring signature scheme is considered unforgeable w.r.t. insider corruption if, for any PPT adversary $\mathcal{A}$, the advantage of $\mathcal{A}$ in the following experiment is at most negligible.

\begin{itemize}
    \item {\rm $\textbf{Setup.}$} The experiment generates {\rm $\textsf{PP}\leftarrow \textsf{Setup}(1^n;\rho_\textsc{st})$} and $(\textsf{vk}^{(i)},\textsf{sk}^{(i)})\leftarrow \textsf{KeyGen}(\rho_\textsc{kg}^{(i)})$ for $i\in[N]$, where $(\rho_\textsc{st},\rho_\textsc{kg}^{(i)})$ are the randomnesses used in {\rm $\textsf{Setup}$} and {\rm $\textsf{KeyGen}$}, respectively. The experiment sets {\rm $\textsf{S}=(\textsf{vk}^{(1)},\dots,\textsf{vk}^{(N)})$} and initializes two empty sets {\rm\textsf{L}} and {\rm\textsf{C}}. Finally, the experiment sends (\textsf{PP}, $\textsf{S}$, $\rho_\textsc{st}$) to $\mathcal{A}$.

   \textsf{Note that we give $\mathcal{A}$ the randomness $\rho_{\textsc{st}}$ used for the Setup algorithm, which implies that the algorithm is public and does not rely on a trusted setup that may incur concerns on the existence of trapdoors hidden in the output parameters.}

    \item {\rm $\textbf{Probing Phase.}$} $\mathcal{A}$ can adaptively query the following oracles:
    \begin{itemize}\setlength{\itemsep}{2pt} \setlength{\parsep}{0pt} \setlength{\parskip}{2pt}
    	\item Signing oracle {\rm $\textsf{OSign}(\cdot,\cdot,\cdot)$}:\\
	On input a message $\mu\in\mathcal{M}$, a ring {\rm \textsf{R}}, and an index $s\in[N]$ such that {\rm $\textsf{vk}^{(s)}\in \textsf{R}\cap \textsf{S}$}, this oracle returns the signature {\rm $\Sigma\leftarrow\textsf{Sign}(\textsf{sk}^{(s)},\mu,\textsf{R})$} and adds the tuple {\rm$(\mu,\textsf{R},\Sigma)$} to {\rm\textsf{L}}.
	
	\item Corrupting oracle {\rm \textsf{OCorrupt}}$(\cdot)$:\\
	 On input an index {$s\in[N]$} such that {\rm $\textsf{vk}^{(s)}\in \textsf{S}$}, this oracle returns {\rm $\rho_\textsc{kg}^{(s)}$} and adds {\rm$\textsf{vk}^{(s)}$} to {\rm\textsf{C}}.
	 
	 %\textsf{the adversary can adaptively corrupt the honest participants with respect to a ring and obtains their signing keys even the randomness used to generate these keys.}
	 
	 %\textsf{Note that we give $\mathcal{A}$ the randomness $\rho_{\textsc{st}}$ used for the Setup algorithm, which implies the algorithm is public, does not rely on a trusted setup that may incur concerns on the existing of trapdoors hidden in the output parameters.}
	 
	\end{itemize}
	
    \item {\rm $\textbf{Forge.}$} $\mathcal{A}$ outputs a forgery {\rm $(\mu^*,\textsf{R}^*,\Sigma^*)$} and succeeds if {\rm  (1) $\textsf{Ver}(\mu^*,\textsf{R}^*,\Sigma^*)=1$}, {\rm (2) $\textsf{R}^*\subseteq \textsf{S}\setminus \textsf{C}$}, and {\rm (3) $(\mu^*,\textsf{R}^*,\Sigma^*)\notin \textsf{L}$}.
   
\end{itemize}

%The advantage of $\mathcal{A}$ is defined by $\textsf{Adv}_{\textsf{Unf}}(\mathcal{A})=\mathrm{Pr}[\textsf{Exp}_{\textsf{Unf}}(\mathcal{A})=1]$.

\subsection{Anonymity Model}\label{PrivacyModel}\pdfbookmark[2]{Anonymity Model}{Anonymity Model}

In our anonymity model, the adversary is allowed to obtain all the ring members' randomness for key generation, as well as the randomness of the \textsf{Setup} algorithm. Moreover, to capture the realistic attack described in Section \ref{OurMethods}, in the challenge phase, the adversary is allowed to obtain $Q$ signatures for any positive integer $Q$ with respect to the same ring member.

%, instead of only one challenge tuple as prior works.

%The rationale behind our unconditional anonymity definition is straightforward. 

%Our unconditional anonymity definition appears to be similar to the notion of anonymity w.r.t. full-key exposure by Bender et al. \cite{BKM06}. 

% In the \textsf{Challenge} phase, the adversary is allowed to provide a challenge set which includes two signing keys and any amount of tuples consisting of different rings and messages. It exactly captures the description of Rivest et al.~\cite{RST01}, namely, allowing the adversary with access to an exponential amount of chosen-message signatures produced by the same ring member. This consideration is necessary, even though the adversary can compute arbitrary amounts by himself since the adversary has all the ring members' signing keys of. Because if there is only one signature replied from the challenger as previous works~\cite{LAZ19,TKS+19,TSS+18,TSS+20,WZZ18,ZHX+16}, the adversary cannot determine the identity of the real signer since the accumulate advantage is mainly derived from the message-signature tuples that generated by himself. The consideration of ``consist of different rings and messages'' is necessary. As Section~\ref{OurApproach} shows the attack on the unconditional anonymity of prior works, even if the distinguishing advantage is only negligible which still can be exploited by the infinitely powerful adversary no matter where the negligible advantage is derived. 

\noindent$\textbf{Anonymity}$. A ring signature scheme satisfies the unconditional anonymity, if for any adversary $\mathcal{A}$, it holds that $\mathcal{A}$ has at most negligible advantage in the following experiment.
\begin{itemize}

\item $\textbf{Setup}$. The experiment generates {\rm $\textsf{PP}\leftarrow \textsf{Setup}(1^n;\gamma_{\textsc{st}})$} and {\rm $(\textsf{vk}^{(i)},\textsf{sk}^{(i)})\leftarrow \textsf{KeyGen}(\gamma_{\textsc{kg}}^{(i)})$} for all $i\in[N]$, where $N=\mathrm{poly}(n)$ and {\rm $(\gamma_{\textsc{st}}, \{\gamma_{\textsc{kg}}^{(i)}\}_{i\in [N]})$} are randomness used in {\rm $\textsf{Setup}$} and {\rm $\textsf{KeyGen}$}, respectively. Finally, the experiment sets {\rm $\textsf{S}=\{\textsf{vk}^{(i)}\}_{i\in[N]}$} and sends {\rm $(\textsf{PP}$, $\textsf{S},\gamma_{\textsc{st}},\{\gamma_{\textsc{kg}}^{(i)}\}_{i\in [N]})$} to $\mathcal{A}$.

\textsf{Note that we not only give $\mathcal{A}$ the randomness $\gamma_{\textsc{st}}$ used in \textsf{Setup}, but also the randomness $\gamma_{\textsc{kg}}$ used in \textsf{KeyGen}. This means that $\mathcal{A}$ can issue signatures arbitrarily.}
    
\item $\textbf{Challenge}$. The challenge phase has two sub-phases:
\begin{itemize}
    \item $\mathcal{A}$ provides a message $\mu^*$ and two indexes $(s_0^*,s_1^*)$ such that $s_0^*\neq s_1^*$ and $s_0^*,s_1^*\in[N]$. The experiment chooses $b\xleftarrow{_\$}\{0,1\}$.
    
    \item For $l\in[Q]$, $\mathcal{A}$ provides a ring {\rm $\textsf{R}_l^*$} such that {\rm $\textsf{vk}^{(s_0^*)},\textsf{vk}^{(s_1^*)} \in \textsf{S}\cap \textsf{R}_l^*$}, the experiment computes $\Sigma_l^*\leftarrow\textsf{Sign}\big(\textsf{sk}^{(s_b^*)},\mu^*, \textsf{R}_l^*\big)$ and sends $\Sigma_l^*$ to $\mathcal{A}$.
\end{itemize}
    
\item $\textbf{Guess}$. $\mathcal{A}$ outputs a bit $b'$. If $b'=b$, the experiment outputs 1, otherwise 0.
    
\end{itemize}

\section{Lattice Backgrounds}\label{LatticeBackgrounds}\pdfbookmark[1]{Lattice Backgrounds}{Lattice Backgrounds}

In this section, we review some lattice backgrounds: the lattices basics and hardness assumption in Section \ref{sec:LatticesBasics}, the concept of bonsai tree and its associated algorithms in Section \ref{sec:BonsaiTree}, as well as the key-homomorphic evaluation algorithm in Section \ref{sec:KeyHomo}.

\subsection{Lattices and Hardness Assumption}\label{sec:LatticesBasics}\pdfbookmark[2]{Lattices and Hardness Assumption}{Lattices and Hardness Assumption}

We consider lattice problems restricted to compact lattices defined over polynomial rings \cite{LPR10}. These are rings of the form $ R=\mathbb{Z}[X]/(\mathrm{\Phi}_{2n}(X))$ or $ R_q= R/qR$ where $n$ is a power of $2$, $q$ is an integer, and $\mathrm{\Phi}_{2n}(X)=X^n+1$ is the cyclotomic polynomial of degree $n$. The $n$-dimensional Gaussian function $\rho_s: \mathbb{R}^n\rightarrow (0,1]$ is defined as $\rho_s(\mathbf{e})=\exp{(-\pi\cdot \lVert \mathbf{e}/s\rVert^2)}$ for any standard deviation $s>0$ and vector $\mathbf{e}\in\mathbb{R}^{n}$. For any countable set $E\subseteq \mathbb{R}^n$, let $\rho_s(E)=\sum_{\mathbf{e}\in E}\rho_s(\mathbf{e})$. The discrete gaussian distribution $D_{\mathrm{\Lambda},s}$ over a lattice $\mathrm{\Lambda}$ is defined as $D_{\mathrm{\Lambda},s}(\mathbf{e})=\rho_s(\mathbf{e})/\rho_s(\mathrm{\Lambda})$.

\begin{definition}[Ring-SIS Assumption~{\rm\cite{GPV08}}]
Short Integer Solution over Rings problem {\rm {RingSIS}}$_{q,n,m,\beta}$ is, given a row vector $\mathbf{A}\in R_q^{1\times m}$, to find a nonzero vector $\mathbf{x}\in R^{m}$ such that $\mathbf{Ax}=\mathbf{0}\pmod q$ and $\lVert \mathbf{x}\rVert\leq \beta$.
\label{DefofSIS}
\end{definition}

\subsection{Bonsai Tree}\label{sec:BonsaiTree}\pdfbookmark[2]{Bonsai Tree}{Bonsai Tree}
%因为是贡献，应当单独拿出来放到construction前面部分比较好。

%由于Cash等人提出的标准Bonsai Tree是由标准格实例化的，而且其构造和证明均只对于单输入的情况。本节，我们对Bonsai Tree进行推广，即使用ring下的格对其进行重实例化，并在多输入的情况下进行重分析和证明。Since the standard Bonsai Tree proposed by Cash et al. is instantiated by standard lattices, and its construction and proofs are only for the single-input case, in this section, we extend the Bonsai Tree by re-instantiating it with lattices over rings, and perform a re-analysis and proof for the multi-input case.

%As mentioned by \cite{CHKP10}, the cryptographic notion of the bonsai tree is a lattice-based cryptographic structure that can be used to construct hash-and-sign signature schemes and hierarchical identity-based encryption schemes in the standard model. In this section, we generalize it to construct ring signatures. Specifically, we show how to generalize the bonsai tree to ring setting, i.e., re-instantiated by the lattices over rings. As the standard bonsai tree \cite{CHKP10}, the generalized bonsai tree is also composed of algorithms: the trapdoor generation algorithm \textsf{TrapGen}, .

% This is motivated by the fact that the ring signature size will be 

% Concretely, the 
% standard bonsai tree proposed by \cite{CHKP10} is instantiated by standard lattices

The concept of the bonsai tree is a lattice-based cryptographic structure was introduced by Cash et al. \cite{CHKP10}. In this work, we employ the bonsai tree instantiated by Micciancio and Peikert's work \cite{MP12}, which has the following associated algorithms and properties.

%In the initial version \cite{CHKP10}, the bonsai tree is instantiated from standard lattices, which mainly consists of two algorithms, the trapdoor generation algorithm \textsf{TrapGen} and the trapdoor delegation algorithm \textsf{TrapDel}\footnote{The trapdoor delegation mechanism in \cite{CHKP10} is consisted of two subalgorithms: a trapdoor extend algorithm and a randomization algorithm (cf. \cite[Section 3, full version]{CHKP10})}. Then, Micciancio and Peikert \cite{MP12} presented an instantiation from lattices over rings, which are as follows:

\begin{lemma}[{\rm \textsf{TrapGen}} Algorithm~{\rm\cite{MP12}}]
For a positive integer $k$, a modulus $q=3^k$, and integer dimension $n$. There is a PPT trapdoor generation algorithm {\rm \textsf{TrapGen}$(q,n,\sigma_{\textsc{tg}})$} that on input $q$, $n$, and a parameter {\rm$\sigma_{\textsc{tg}}>\omega(\sqrt{\ln{nw}})$}, outputs a trapdoor $\mathbf{T}\in R^{w\times k}$ for  $\mathbf{A}\in R_q^{1\times (k+w)}$ such that {\rm$s_1(\mathbf{T})\leq\sigma_{\textsc{tg}}\cdot O(\sqrt{w}+\sqrt{k}+\omega(\sqrt{\log n}))$}. Moreover, if $w\geq 2(\lceil \log q\rceil + 1)$ then with overwhelming probability the distribution of $\mathbf{A}$ is statistically close to uniform.
\label{TrapGen}
\end{lemma}

% \begin{lemma}[{\rm \textsf{TrapGen}} Algorithm~{\rm\cite{MP12}}]
% For a positive integer $k$, a modulus $q=3^k$, and integer dimension $n$, define the Gadget matrix $\mathbf{G}=[1,3,9,\dots,3^{k-1}]\in R^{1\times k}$. There is a PPT trapdoor generation algorithm {\rm \textsf{TrapGen}}$(\bar{\mathbf{A}},\mathbf{H},s)$ that on input a matrix $\bar{\mathbf{A}}\in R_q^{1\times w}$, invertible matrix $\mathbf{H}\in R_q$, and a parameter $s>\omega(\sqrt{\ln{nw}})$, outputs a trapdoor $\mathbf{T}\in R^{w\times k}$ for  $\mathbf{A}\in R_q^{1\times (k+w)}$ with tag $\mathbf{H}$ such that $s_1(\mathbf{T})=s\cdot O(\sqrt{w}+\sqrt{k}+\omega(\sqrt{\log n}))$. Moreover, if $w\geq 2(\lceil \log q\rceil + 1)$ then with overwhelming probability over the random choice of $\bar{\mathbf{A}}$, the distribution of $\mathbf{A}$ is statistically close to uniform.
% \label{TrapGen}
% \end{lemma} 

\begin{lemma}[{\rm \textsf{TrapDel}} Algorithm {\rm\cite{MP12}}] There is a PPT trapdoor delegation algorithm {\rm $\textsf{TrapDel}(\mathbf{T}, \mathbf A'=[\mathbf{A}\mid \mathbf{A}_1], \sigma_{\textsc{td}})$} that on input a concatenate matrix $\mathbf A'=[\mathbf{A}\mid \mathbf{A}_1]$ where $\mathbf{A}\in R_q^{1\times (k+w)}$ and $\mathbf{A}_1$ is an arbitrary matrix in $ R_q^{1\times w}$, a trapdoor $\mathbf{T}\in R^{w\times k}$ of $\mathbf{A}$,  and a parameter {\rm$\sigma_{\textsc{td}}\geq s_1(\mathbf{T})\cdot \omega(\sqrt{\log n})$}, outputs a delegated trapdoor $\mathbf{T}'\in R^{(k+w)\times k}$ of $\mathbf{A}'$ such that {\rm$s_1(\mathbf{T}')\leq \sigma_{\textsc{td}} \cdot O(\sqrt{k+w}+\sqrt{k}+\omega(\sqrt{\log n}))$} and $\mathbf{T}'$ is distributed statistically independent with $\mathbf{T}$. 
\label{TrapDel}
\end{lemma}

\begin{lemma}[Trapdoor Indistinguishability of {\rm \textsf{TrapDel}}] Let $\mathbf{T}'_0, \mathbf{T}'_1$ be any two delegated trapdoors generated by {\rm $\textsf{TrapDel}(\mathbf{T}, \mathbf A'_0, \sigma_{\textsc{td}})$} and {\rm $\textsf{TrapDel}(\mathbf{T}, \mathbf A'_1, \sigma_{\textsc{td}})$}, respectively, i.e., $\mathbf{T}'_0, \mathbf{T}'_1$ are two trapdoors for the same $\mathbf A'$, the $\mathbf{T}'_0$ and $\mathbf{T}'_1$ are within {\rm negl$(n)$} statistical distance. Furthermore, it is also holds for the case $\mathbf A'_0 = \mathbf A'_1$.
\label{Lem:TrapDelProperty}
\end{lemma}

% \begin{lemma}[{\rm \textsf{TrapDel}} Algorithm {\rm\cite{MP12}}] There is a PPT lattice trapdoor delegation algorithm {\rm $\textsf{TrapDel}(\mathbf{T}, \mathbf A'=[\mathbf{A}\mid \mathbf{A}_1], \mathbf{H}', \sigma)$} that on input a concatenate matrix $\mathbf A'=[\mathbf{A}\mid \mathbf{A}_1]$ where $\mathbf{A}\in R_q^{1\times (k+w)}$ and $\mathbf{A}_1$ is an arbitrary matrix in $ R_q^{1\times w}$, a trapdoor $\mathbf{T}\in R^{w\times k}$ of $\mathbf{A}$, an invertible matrix $\mathbf{H}'\in R_q$, and a parameter $\sigma\geq s_1(\mathbf{T})\cdot \omega(\sqrt{\log n})$, outputs a trapdoor $\mathbf{T}'\in R^{2w\times k}$ of $\mathbf{A}'$ such that $s_1(\mathbf{T}')\leq \sigma \cdot O(\sqrt{w}+\sqrt{k})$. 
% \label{TrapDel}
% \end{lemma}

We note that the Trapdoor Indistinguishability property of \textsf{TrapDel} will be useful in our anonymity proof (cf. the proof of Theorem \ref{Them:AnonymityProof}). This is achieved by the underlying Gaussian sampling algorithm \textsf{GauSample}. We also use the \textsf{GauSample} algorithm to sample signatures in the \textsf{Sign} algorithm.

\begin{lemma}[{\rm \textsf{GauSample}} Algorithm~{\rm\cite{DM14}}]
There is a PPT Gaussian sampling algorithm {\rm $\textsf{GauSample}(\mathbf A, \mathbf{T},\mathbf u,\sigma_{\textsc{Saml}})$} that on input a matrix $\mathbf{A}\in R_q^{1\times (k+w)}$, a trapdoor $\mathbf{T}\in R^{w\times k}$ of $\mathbf{A}$, a syndrome $\mathbf{u}\in R_q$, and a sufficiently large parameter {\rm$\sigma_{\textsc{Saml}}\geq s_1(\mathbf{T})\cdot\omega(\sqrt{\log n})$}, the algorithm samples a preimage $\mathbf{e}\in R^{k+w}$ which is distributed statistically independent with the trapdoor $\mathbf{T}$ and satisfies that {\rm$\mathbf{Ae} = \mathbf{u}\pmod q$ and $\lVert \mathbf{e}\rVert\leq \sigma_{\textsc{Saml}}\sqrt{k+w}$} holds with overwhelming probability.
\label{GauSample}
\end{lemma}

\subsection{Key-Homomorphic Evaluation Algorithm}\label{sec:KeyHomo}\pdfbookmark[2]{Key-Homomorphic Evaluation Algorithm}{KeyHomo}

In our construction, the key-homomorphic evaluation algorithm $\textsf{Eval}(\cdot, \cdot)$ that developed from the works~\cite{BGG+14,BV14}, is employed in the \textsf{Sign} algorithm to effectively reduce the ring signature size. Additionally, in the simulation, we reconstruct the matrix $\mathbf A^{(i)}$ by the gadget matrix $\mathbf{G}$ and the random matrix $\mathbf R^{(i)}$, that is\[\mathbf A^{(i)}=\mathbf A\mathbf R^{(i)}+b^{(i)}\mathbf G\in R_q^{1\times m}\]\[\mathbf{G}=[1,3,9,\dots,3^{m-1}]\in R^{1\times m},\quad\ \mathbf R^{(i)}\xleftarrow{\$}\{1,-1\}^{m\times m}\subset R_q^{(m/n) \times (m/n)}\]where $m=k+w$. In this setting, we employ the key-homomorphic evaluation algorithm to implicitly 
evaluate a PRF function by taking these reconstructed matrices into a {\rm $\textsf{NAND}$} Boolean circuit. Below we review a fact of the key-homomorphic evaluation algorithm which shows the implicit evaluated matrix $\mathbf R_{C}$ has a low-norm, and also review the PRF definition.

% In our construction, we utilize the key-homomorphic evaluation algorithm $\textsf{Eval}(\cdot, \cdot)$, which was developed in~\cite{BGG+14,BV14}, in the \textsf{Sign} algorithm to effectively reduce the ring signature size. In simulations, we reconstruct the matrix $\mathbf A^{(i)}$ using the gadget matrix $\mathbf{G}$ and a random matrix $\mathbf R^{(i)}$, as follows: [\mathbf A^{(i)}=\mathbf A\mathbf R^{(i)}+b^{(i)}\mathbf G\in R_q^{1\times m},] [\mathbf{G}=[1,3,9,\dots,3^{m-1}]\in R^{1\times m},\quad\ \mathbf R^{(i)}\xleftarrow{$}{1,-1}^{m\times m}\subset R_q^{(m/n) \times (m/n)},] where $m=k+w$. In this setting, the key-homomorphic evaluation algorithm is employed to implicitly compute a pseudorandom function (PRF) by taking these reconstructed matrices into a {\rm $\textsf{NAND}$} Boolean circuit. Here we review a key fact about the key-homomorphic evaluation algorithm: the implicitly evaluated matrix $\mathbf R_{C}$ has alow-norm. We also review the definition of a PRF.

\begin{lemma}[{\rm\cite{BGG+14,BV14}}]
Let $d=c \log \ell$ for some constant $c$. Define $C:\{0,1\}^\ell \rightarrow \{0,1\}$ be a {\rm $\textsf{NAND}$} boolean circuit with depth $d$. Let $\{\mathbf A^{(i)}=\mathbf A\mathbf R^{(i)}+b^{(i)}\mathbf G\}_{i\in[\ell]}$ be $\ell$ different matrices correspond to each input wire of ${C}$ where $\mathbf A\xleftarrow{_\$} R_q^{1\times m}$, $\mathbf R^{(i)}\xleftarrow{_\$} \{1,-1\}^{m\times m}$, $b^{(i)}$ is a bit over $R$. The algorithm {\rm $\textsf{Eval}(C, (\mathbf A^{(1)},\cdots, \mathbf A^{(\ell)}))$} runs in time $\mathrm{poly}(4^d,\ell,n,\log q)$, the inputs are $C$ and $\{\mathbf A^{(i)}\}_{i\in[\ell]}$, the output is 
\begin{center}
{\rm $\mathbf A_C\ =\ \mathbf A\mathbf R_{C}+{C}(b^{(1)},\dots,b^{(\ell)})\cdot\mathbf G\ =\ \textsf{Eval}(C, (\mathbf A^{(1)},\dots,\mathbf A^{(\ell)}))$}
\end{center}
 where ${C}(b^{(1)},\dots,b^{(\ell)})$ is the output bit of $C$ on the arguments $(b^{(1)},\dots,b^{(\ell)})$ and $\mathbf R_{C}\in R^{m\times m}$ is a low norm matrix has $\lVert\mathbf R_C\rVert \leq O(\ell\log m + \sqrt{m})$. 
\label{Lemma7}
\end{lemma}

\begin{definition}[Pseudorandom Functions]\label{DefofPRF}
Let $\kappa$, $t$, and $c$ be polynomial in $n$. A PRF function {\rm $\textsf{PRF}:\{0,1\}^\kappa\times \{0,1\}^t\rightarrow \{0,1\}^c$} is a deterministic two-input function where the first input, denoted by $K$, is the key. Let $\Omega$ be the set of all functions whose domain is $t$ bits strings while the range is $c$ bits strings, then {\rm$$\big{\lvert} \mathrm{Pr}\big[\mathcal{A}^{\textsf{PRF}(K,\cdot)}(1^n)=1\big]-\mathrm{Pr}\big[\mathcal{A}^{F(\cdot)}(1^n)=1\big]\big{\rvert} \leq \textsf{negl}(n)$$}where the probability is taken over a uniform choice of key $K\xleftarrow{_\$}\{0,1\}^\kappa$, $F\xleftarrow{_\$} \Omega$, and the randomness of $\mathcal{A}$. 
\end{definition}

\section{The Ring Signature Scheme}\label{TheRingSignatureScheme}\pdfbookmark[1]{The Ring Signature Scheme}{The Ring Signature Scheme}

In this section, we present a ring signature scheme from the lattices over rings. The construction is presented in Section \ref{Construction}, the correctness and parameters setting are given in Section \ref{Correctness}, the implementation is given in Section \ref{Sec:Implementation}, and the unforgeability and anonymity is proven in Appendix \ref{Appen:Proofs}.

% Our RS employs two methods: The Bonsai-Tree framework from \cite{CHKP10} and the key-homomorphic evaluation algorithm from \cite{BL16}. Even though both of them are from the standard signature schemes, the essential differences are induced when employed in RS setting since the syntax and security requirements all changed significantly. At a high level, we use the key-homomorphic evaluation algorithm $\textsf{Eval}(\cdot,\cdot)$ to efficiently ``process'' the message while using the Bonsai-Tree mechanism (include algorithms \textsf{BasisExd}, \textsf{RandBasis}, and \textsf{GauSample}) to sampling the signature. In the key generation phase, each ring member needs to generate enough matrices: $(\mathbf A_0, \mathbf A_1)$ are used to match the bit $d$ of PRF result, $\{\mathbf B_{j}\}_{j\in[k]}$ are used to match the bits of PRF key $\mathbf{k}$ while the $(\mathbf C_0, \mathbf C_1)$ are used to match the message bits of $\boldsymbol{\mu}$. When signing messages, the computation and evaluation on these matrices are required for each ring member with index $i\in[N]$, the resulting matrix is $\mathbf F_{\boldsymbol{\mu},1-d}^{(i)}$. Then the signer constructs an augmented matrix by concatenating each $\mathbf F_{\boldsymbol{\mu},1-d}^{(i)}$, and computes a randomized and extended basis from the signer's basis by the algorithms \textsf{BasisExd} and \textsf{RandBasis}. Finally, a preimage vector $\mathbf{e}$ is sampled as the signature by \textsf{GauSample}.

\subsection{Construction}\label{Construction}\pdfbookmark[2]{Construction}{Construction}

\noindent$\textsf{Setup}(1^n):$

\begin{enumerate}\setlength{\itemsep}{3pt} \setlength{\parsep}{0pt} \setlength{\parskip}{1pt}
	\item On input the security parameter $n$, sets the modulo $q$, lattice dimension $m$, PRF key length $\kappa$, message length $t$, and Gaussian parameters {\rm$\sigma_{\textsc{tg}}$}, {\rm$\sigma_{\textsc{td}}$}, and {\rm$\sigma_{\textsc{Saml}}$} as specified in Section \ref{Correctness}.	
	\item Select a secure $\textsf{PRF}:\{0,1\}^\kappa\times \{0,1\}^t \rightarrow \{0,1\}$, express it as a $\textsf{NAND}$ boolean circuit ${C_\textsf{PRF}}$.
	\item Output the public parameters {\rm$\textsf{PP}=(q,m,\kappa,t,\sigma_{\textsc{tg}},\sigma_{\textsc{td}},\sigma_{\textsc{Saml}},\textsf{PRF})$}.
\end{enumerate}
In the following, $\textsf{PP}$ are implicit input parameters to every algorithm.

\vspace{3mm}
\noindent$\textsf{KeyGen}():$

\begin{enumerate} \setlength{\itemsep}{2pt} \setlength{\parsep}{0pt} \setlength{\parskip}{1pt}
	
	\item Compute $(\mathbf A,\mathbf T)\leftarrow \textsf{TrapGen}(q,n,\sigma_{\textsc{tg}})$.

	\item Select $\mathbf A_0, \mathbf A_1, \mathbf C_0, \mathbf C_1 \xleftarrow{_\$}R_q^{1\times m}$ and $\mathbf u\xleftarrow{_\$}R_q$.
	
	\item Select a $\textsf{PRF}$ key $\mathbf k=(k_1,\dots,k_\kappa)\xleftarrow{_\$}\{0,1\}^\kappa$.

	\item For $j=1$ to $\kappa$, select $\mathbf B_{j}\xleftarrow{_\$} R_q^{1\times m}$.
	
	\item Output $\textsf{vk}=(\mathbf A, (\mathbf A_{0},\mathbf A_{1}), \{\mathbf B_{j}\}_{j\in[\kappa]}, (\mathbf C_{0},\mathbf C_{1}), \mathbf u)$ and $\textsf{sk}=(\mathbf T,\mathbf k)$.
\end{enumerate}

In the rest of the construction, for a ring $\textsf{R}=(\textsf{vk}^{(1)},\dots,\textsf{vk}^{(N)})$, we implicitly parse each verification key $\textsf{vk}^{(i)}=(\mathbf A^{(i)},(\mathbf A_{0}^{(i)},\mathbf A_{1}^{(i)}), \{\mathbf B_{j}^{(i)}\}_{j\in[\kappa]},(\mathbf C_{0}^{(i)}, \mathbf C_{1}^{(i)}), \mathbf u^{(i)})$, and the corresponding signing key $\textsf{sk}^{( i)}=(\mathbf T^{(i)}, \mathbf{k}^{(i)})$.

\vspace{2mm}
\noindent$\textsf{Sign}(\boldsymbol{\mu},\textsf{R},\textsf{sk}):$
\begin{enumerate} \setlength{\itemsep}{3pt} \setlength{\parsep}{0pt} \setlength{\parskip}{1pt}

	\item On input a message $\boldsymbol{\mu}=(\mu_1,\dots,\mu_t)\in\{0,1\}^t$, a ring of verification keys $\textsf{R}$, and a signing key $\textsf{sk}^{(s)}$ where $s \in [N]$ is the index of the signer in ring $\textsf{R}$.
	
% 	\item Verify that the implicit input public parameters \textsf{PP} satisfy the setting given in Section \ref{Correctness}; $\mathbf S_{\mathbf A^{(s)}}$ is in $\mathbb{Z}^{m\times m}$ such that $\lVert \widetilde{\mathbf S_{\mathbf A^{(s)}}} \rVert\leq O(\sqrt{n \log q})$; $N$ is polynomial in $n$; If the verification fails output \textsf{failed}.
		
	\item Compute $d=\textsf{PRF}(\mathbf k^{(s)}, \boldsymbol{\mu})$.
	
	\item For $i\in[N]$, compute $\mathbf A_{\boldsymbol{\mu}}^{(i)}=\textsf{Eval}(C_\textsf{PRF},(\{\mathbf B_{j}^{(i)}\}_{j\in[\kappa]},\mathbf C_{\mu_1}^{(i)},\mathbf C_{\mu_2}^{(i)},\dots,\mathbf C_{\mu_t}^{(i)}))$, then set $\mathbf F_{\boldsymbol{\mu},1-d}^{(i)}=\big[\mathbf A^{(i)}\mid \mathbf A_{1-d}^{(i)}-\mathbf A_{\boldsymbol{\mu}}^{(i)}\big]\in R_q^{1\times 2m}$.	
	
	\item Let $\mathbf F'_{\boldsymbol{\mu},1-d}=\big[\mathbf F_{\boldsymbol{\mu},1-d}^{(1)}\mid\dots\mid\mathbf F_{\boldsymbol{\mu},1-d}^{(N)}\big]\in R_q^{1\times 2Nm}$. Delegate the trapdoor $\mathbf{T}_{\mathbf F'_{\boldsymbol{\mu},1-d}}$ for $\mathbf F'_{\boldsymbol{\mu},1-d}$ by 
	\[
	    \mathbf{T}_{\mathbf F'_{\boldsymbol{\mu},1-d}}\leftarrow\textsf{TrapDel}\big(\mathbf{T}^{(s)}, \mathbf F'_{\boldsymbol{\mu},1-d},  \sigma_{\textsc{td}}\big).
	\]
	%\vspace{-4mm}
	\item Let $i'$ be the index for which $\mathbf A^{(i')}$ is lexicographically first\footnote{The signer-anonymity is immediately broken if we use the $\mathbf{u}^{(s)}$ of the signer. To avoid that, we select the parameter $\mathbf{u}^{(i')}$ in the lexicographically first way as the works~\cite{BK10,PS19}.} in $\{\mathbf A^{(i)}\}_{i\in[N]}$.

	\item Sample $\mathbf e=(\mathbf e^{(1)},\dots,\mathbf e^{(N)})$ by $\textsf{GauSample}(\mathbf F'_{\boldsymbol{\mu},1-d},\mathbf T_{\mathbf F'_{\boldsymbol{\mu},1-d}}, \mathbf u^{(i')},\sigma_{\textsc{Saml}})$.
	
	\item Output the signature $\Sigma=\mathbf e$.

\end{enumerate}

\vspace{3mm}
\noindent$\textsf{Ver}(\boldsymbol{\mu},\textsf{R},\Sigma):$

\begin{enumerate}\setlength{\itemsep}{3pt} \setlength{\parsep}{0pt} \setlength{\parskip}{1pt}
	\item On input a message $\boldsymbol{\mu}$, a ring of verification keys $\textsf{R}$, and a signature $\Sigma=\mathbf{e}$, parse $\mathbf e=(\mathbf e^{(1)},\dots,\mathbf e^{(N)})$.
	
	\item For $i\in[N]$ and $d\in\{0,1\}$, set $\mathbf F_{\boldsymbol{\mu},d}^{(i)}=\big[\mathbf A^{(i)} \mid \mathbf A_{d}^{(i)}-\mathbf A_{\boldsymbol{\mu}}^{(i)}\big]$ where $\mathbf A_{\boldsymbol{\mu}}^{(i)}$ is computed as in \textsf{Sign} algorithm.
		
	\item Select $\mathbf u^{(i')}$ from $\{\mathbf u^{(i)}\}_{i\in[N]}$ for which the corresponding $\mathbf A^{(i')}$ is lexicographically first in $\{\mathbf A^{(i)}\}_{i\in[N]}$.

	\item Let $\mathbf F'_{\boldsymbol{\mu},d}=\big[\mathbf F_{\boldsymbol{\mu},d}^{(1)}\mid\dots\mid\mathbf F_{\boldsymbol{\mu},d}^{(N)}\big]$, check if each $\lVert \mathbf e^{(i)}\rVert\leq \sigma_{\textsc{Saml}}\sqrt{m}$ for $i\in[N]$, and $\mathbf F'_{\boldsymbol{\mu},d}\cdot \mathbf e=\mathbf u^{(i')}\pmod q$ holds for $d=0$ or $1$, accept the signature; otherwise, reject\footnote{The setting of `$d=0$ or $1$' and the difference that we use $\bold F_{\boldsymbol{\mu},1- d}^{(i)}$ in \textsf{Sign} algorithm while $\bold F_{\boldsymbol{\mu},d}^{(i)}$ in \textsf{Ver}, are for the unforgeability proof. (cf. proof of Theorem \ref{Them:UnforgeabilityProof} for details).}.

 %Specifically, these settings can prevent the adversary from trivially forging the signatures by using the public trapdoor of a so-called gadget matrix $\mathbf{G}$ 
\end{enumerate}

\subsection{Correctness and Parameters}\label{Correctness}\pdfbookmark[2]{Correctness and Parameters}{Correctness and Parameters}
The correctness of the scheme is easily verified: By Lemma \ref{GauSample}, the vector $\mathbf{e}$ satisfies that $\mathbf F'_{\boldsymbol{\mu},d}\cdot \mathbf e=\mathbf u^{(i')}\pmod q$ holds for $d=0$ or $1$, and the length of each $\mathbf e^{(i)}$ is at most $\sigma_{\textsc{Saml}}\sqrt{m}$ with overwhelming probability. Therefore, the signature can be accepted by the \textsf{Ver} algorithm. We then explain the parameters choosing. 

We parameterized the ring signature scheme by the security parameter $n$ which we assume is a power of $2$, and a modulus $q=3^k$ which we assume to be a power of $3$. These parameters define the ring $R_q=\mathbb{Z}[X]/(\mathrm{\Phi}_{2n}(X),q)$ where $\mathrm{\Phi}_{2n}(X)=X^n+1$ is the cyclotomic polynomial of degree $n$. For the PRF, we instantiate it by the work \cite{BPR12} in which the $\mathrm{poly}(n)$-bounded modulus $q=\mathrm{poly}(n)$ and key length $\kappa=\mathrm{poly}(n)$ are allowed. Let $\ell=t+k$ be the input length of PRF where $t$ is the message length. The scheme also uses the following parameters:

$$\sigma_{\textsc{tg}}=\omega(\sqrt{\ln{nw}}),\quad~  w=2\lceil \log q\rceil +2\quad $$ $$\sigma_{\textsc{Saml}}=O(\sqrt{Nk\ell}\log m + \sqrt{Nkm})\cdot \omega({\log n})$$ $$\ \  \sigma_{\textsc{td}}= s_1(\mathbf{T})\cdot \omega(\sqrt{\log n}),\quad \beta=O(m^{3/2})\cdot \sigma_{\textsc{tg}}\sigma_{\textsc{Saml}}$$where the $(\sigma_{\textsc{tg}},w)$ and $\sigma_{\textsc{td}}$ is given by Lemma \ref{TrapGen} and \ref{TrapDel}, respectively. The $\sigma_{\textsc{Saml}}$ and $\beta$ is given by our unforgeability proof (cf. the proof of Theorem \ref{Them:UnforgeabilityProof}). Specifically, in order to guarantee the distribution on the output of $\textsf{GauSample}$ statistically indistinguishable between the real and simulated world, we need to set the Gaussian parameter $\sigma_{\textsc{Saml}}$ sufficiently large; To ensure the preimage vector $\mathbf{e}$ sampled by $\textsf{GauSample}$ are not trivial to find (otherwise, the unforgeability is broken), we need to set a norm bound, i.e., the parameter $\beta$ on the preimage vector. In the following implementation, we give a concrete parameters setting.

%To ensure our construction has a worst-case lattice reduction, by the result that given by Gentry et al. \cite{GPV08}, it is required to set $q=\beta\cdot \omega(\sqrt{n\log n})$ (at least). 

%我们采用了G-trapdoor，该trapdoor的生成算法简单高效且所输出的trapdoor尺寸较小，且与该trapdoor相关联的G-SamplePre算法在支持并行执行的同时还能够将对A高斯采样的主要计算量转移到G上，其中G是一个特殊的器件矩阵。更重要的是，我们巧妙地将Dagdelen等人所提出的用于标准格基签名方案的优化技术（尤其是签名算法中的rearrange operations）与G-SamplePre的线上/线下功能相结合，使得签名算法的大部分的中间存储和计算量可转移至线下并以预计算的方式执行。 在具体实施上，我们克服了格基的高存储的缺点以及陷门采样算法的高
% On the concrete implementation, we overcome the costly Gaussian sampling procedure by combining the functionalities (parallel preimage sampling and online/offline setting) of G-\textsf{SamplePre} algorithm \cite{MP12} with the optimization techniques (especially the rearrange operations in signing algorithm) for signature schemes from standard lattices by Dagdelen et al. \cite{DBG+14}. For the high storage overhead of the lattice basis, we resolve it by employing the G-trapdoor \cite{MP12} and its associated algorithms G-\textsf{BasisGen}, which efficiently delivers lattice basis, smaller than general basis, since it essentially amounts to just one multiplication of two random matrices.

% We implement our scheme using $\mathrm{C}^{++}$-language, based on the Palisade library \cite{Palisade}. Table \ref{tab:Tab1} and Table \ref{tab:Tab2} shows the experimental results of our implementation on a usual computation environment, namely, a desktop with Intel(R) Core(TM) i7-6500U CPU @2.5GHz., 16 GB memory, and the operating system Ubuntu 16.04 LTS. 

\subsection{Implementation}\label{Sec:Implementation}\pdfbookmark[2]{Implementation}{Implementation}

%Our implementation will show the efficiency on the concrete computation time (\textsf{KeyGen}, \textsf{Sign}, and \textsf{Ver}) and parameters size (\textsf{vk}, \textsf{sk}, and signature size). Our implementation will also show the proposed ring signature scheme immune the realistic attack as shown in Section \ref{OurMethods}. 

To generate our specific discrete Gaussian distributions, we employ two building blocks: the AES-based pseudorandom number generator from \cite{MN17} which is implemented using AES-NI instructions for x86 architectures, and the DM-sampler \cite{DM14} which can generate samples in constant time and these samples are distributed independently with the input trapdoor and the sampled preimages. On the concrete implementation, we implement our ring signature scheme using C++ language in the 3.4GHz Intel(R) Core(TM) i7-6600U CPU, 16GB Memory, and the operating system ubuntu 20.04 LTS, and based on the NFLlib library which is available at \url{https://github.com/quarkslab/NFLlib}.

\emph{Concrete Parameters}. Gaussian sampling is the most basic sub-algorithm of our bonsai tree mechanism, according to the employed DM-sampler \cite{DM14}, we use the computation instantiation of the sampler with $n=512$ and $k=30$. Therefore, we have $q=3^{k}=3^{30}$, by Lemma \ref{TrapGen}, $w=2\lceil \log q\rceil +2=98$, then $m=w+k=128$. To guarantee the hardness of the {\rm {RingSIS}}$_{q,n,m,\beta}$ instance, we follow the general framework of \cite{BFL+18}. Specifically, to ensure that the shortest vector outputted by BKZ is a solution of {\rm {RingSIS}}$_{q,n,m,\beta}$ instance with norm bound $\beta$, the root Hermite factor $\delta$ should satisfy that $\frac{\beta}{\sqrt{q}}=\delta^{2n}$.

To minimize the computational cost of the cryptographic operations involved in PRF, we set $N=2$, $\kappa=128$, $t=1$, and $\ell=\kappa+t=129$. According to the parameters relations given in Section \ref{Correctness}, we can calculate the parameters $\sigma_{\textsc{Saml}}$ and $\beta$. Table \ref{tab:ConcreteParametersSetting} shows the concrete parameters setting.

\begin{table}\renewcommand\arraystretch{1.5}
\centering
\caption{Concrete parameters of our ring signaure scheme.}
\label{tab:Tab1}
\begin{tabular}{c||c|c|c|c|c||c|c|c|c}
\hline
$\delta$  &  $n$ & $k$ & $w$ & $m$ &\ $q$ \ & $\sigma_{\textsc{tg}}$ & $\sigma_{\textsc{td}}$ & \begin{tabular}[c]{@{}l@{}} $\sigma_{\textsc{Saml}}$ \end{tabular} \ & \ \begin{tabular}[c]{@{}l@{}} $\beta$\end{tabular}\ \  \\ \hline
\hspace{0.6mm} $1.002985$\hspace{0.6mm} &\hspace{0.6mm} $512$ \hspace{0.6mm} & \hspace{0.6mm} $30$ \hspace{0.6mm} & \hspace{0.6mm} $98$ \hspace{0.6mm} & \hspace{0.6mm} $128$ \hspace{0.6mm} &\hspace{0.6mm}  $3^{30}$ \hspace{0.6mm} & \hspace{0.6mm} $3.3$ \hspace{0.6mm} & \hspace{0.6mm} $55.1$ \hspace{0.6mm} & \hspace{0.6mm} $63740.1$ \hspace{0.6mm} & \hspace{0.6mm}  $303,684,288.2$ \hspace{0.6mm}   \\ \hline
\end{tabular}
\label{tab:ConcreteParametersSetting}
\end{table}

\vspace{2mm}
\emph{Benchmark of performance}. We benchmark the concrete computation time of algorithms \textsf{KeyGen}, \textsf{Sign}, and \textsf{Ver}) and signature size with varying ring size (cf. Figure \ref{OurPerformance}). From Figure \ref{OurPerformance:b}, we can observe the experimental result is better since the signature size is sublinear with the ring size rather than linear as our scheme theoretically showed. The reason is the signature size is actually dominated by the parameter $\sigma_{\textsc{Saml}}$ in which the ring size is $\sqrt{N}$ rather than $N$.

\begin{figure}[!htb]
    \ffigbox[\textwidth]
    {
        \begin{subfloatrow}[2]%useFCwidth
        \ffigbox[\FBwidth]{
            \includegraphics[width=5cm,height=3.8cm]{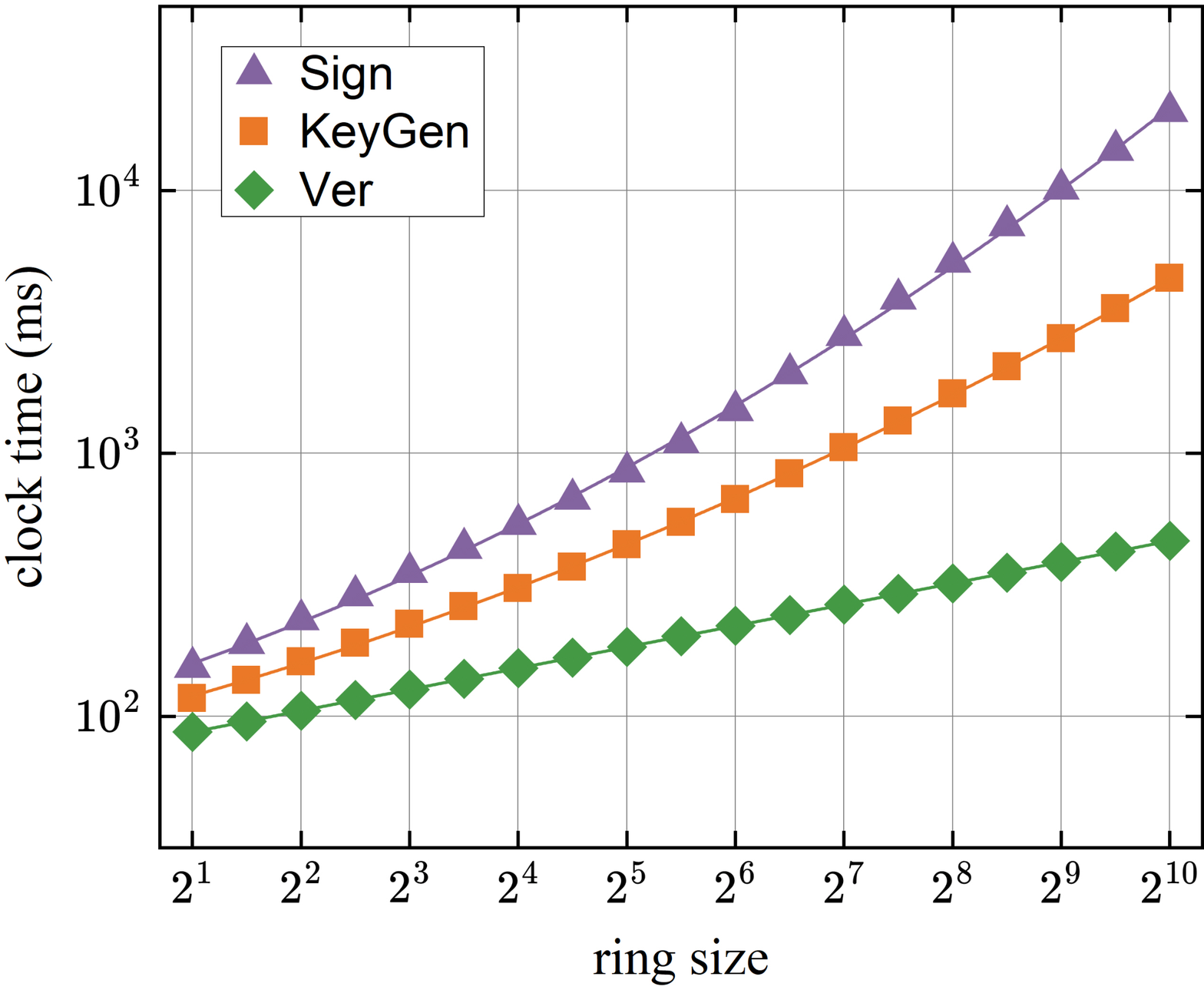}
        }{\caption{Clock time for \textsf{KeyGen}, \textsf{Sign}, and \textsf{Ver} with varying ring size.}\label{OurPerformance:a}}
        \ffigbox[\FBwidth]{
            \includegraphics[width=5cm,height=3.8cm]{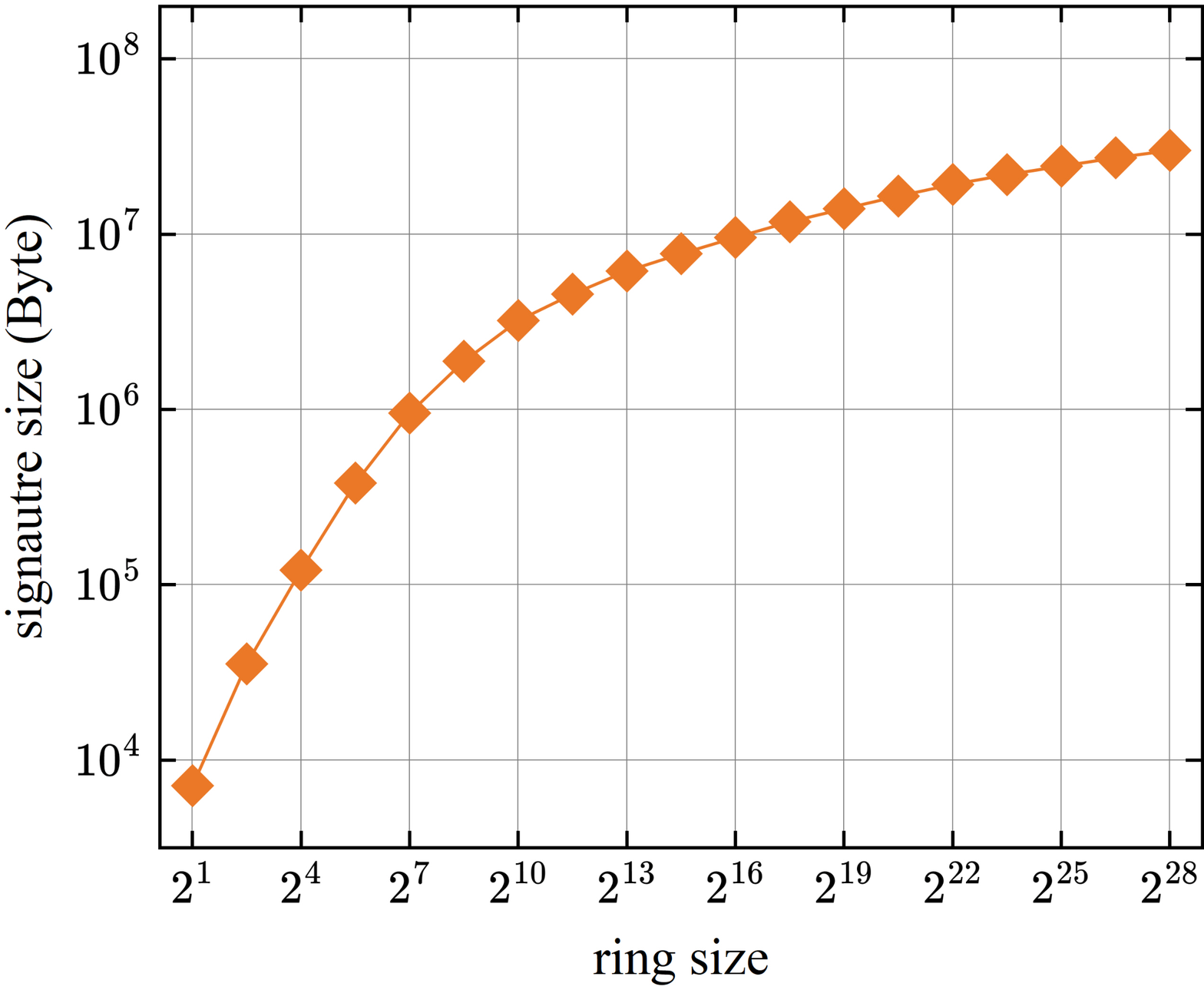}
        }{\caption{Concrete signature size with varying ring size.}\label{OurPerformance:b}}
        \end{subfloatrow}   
    }{\caption{Concrete computation time and signature size of our ring signature scheme.}\label{OurPerformance}}
    \end{figure}

\begin{figure}[!htb]
    \ffigbox[\textwidth]
    {
        \begin{subfloatrow}[2]%useFCwidth
        \ffigbox[\FBwidth]{
            \includegraphics[width=5.1cm,height=3.8cm]{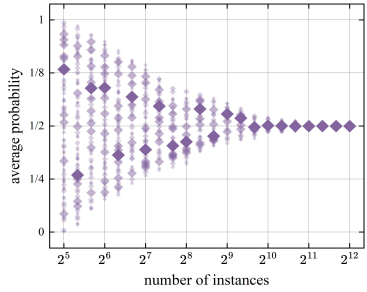}
        }{\caption{Evaluate the average probability for our ring signature with $|\textsf{R}|=2$.}\label{AttackOurRS:a}}
        \ffigbox[\FBwidth]{
            \includegraphics[width=5.1cm,height=3.8cm]{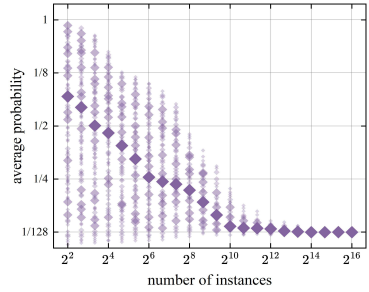}
        }{\caption{Evaluate the average probability for our ring signature with $|\textsf{R}|=128$.}\label{AttackOurRS:b}}
        \end{subfloatrow}    
        \par\nointerlineskip\vspace{13mm}
        \begin{subfloatrow}[2]%useFCwidth      
        \ffigbox[\FBwidth]{
            \includegraphics[width=5.1cm,height=3.8cm]{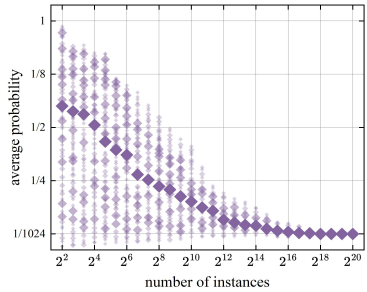}
        }{\caption{Evaluate the average probability for our ring signature with $|\textsf{R}|=1024$.}\label{AttackOurRS:c}}
        \ffigbox[\FBwidth]{
            \includegraphics[width=5.1cm,height=3.8cm]{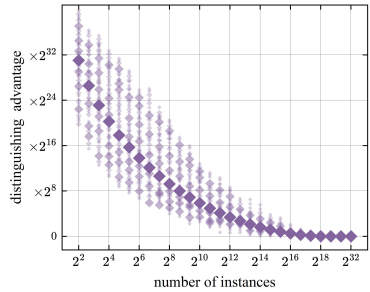}
        }{\caption{Evaluate the distinguishing disadvantage based on the results of Figures \ref{AttackOurRS:a},\ref{AttackOurRS:b},and \ref{AttackOurRS:c}.}\label{AttackOurRS:d}}
        \end{subfloatrow}
    }{\caption{Attack strategy implementation.}\label{AttackOurRS}}
    \end{figure}

\emph{Benchmark of our attack strategy}. 
In our anonymity proof (cf. proof of Theorem \ref{Them:AnonymityProof}), we theoretically proved that the signer-anonymity is hold even the adversary is unbounded, now we examine it in practice. By the attack strategy given in Section \ref{OurMethods}, there is also a non-negligible distinguishing advantage on our scheme, namely, the attack procedure will finally point out one ring member by outputting its index, but which cannot be taken as an evidence to determine the identity of the real signer since the outputted index is distributed uniform over the index distribution. We demonstrate that by the following example.
 
We fix one ring member to generate $Q$ signature sets $\textsf{Sig}_1,\dots,\textsf{Sig}_Q$ independently, let $l\in[Q]$, each $\textsf{Sig}_l$ contains enough signature instances for the attack procedure to output a result. The attack procedure takes as input one $\textsf{Sig}_l$ each time, then outputs an index $\textsf{i}^{(i)}$ for $i\in[N]$ (assume the ring size is $N$). Let $\textsf{Ind} = \{\textsf{i}^{(1)},\dots,\textsf{i}^{(N)}\}$ be the set of all output indexes. The goal is to prove the $\textsf{i}^{(i)}$ is distributed uniform over \textsf{Ind}. 

We define $\Delta(X,Y)$ by two variables $X,Y$ which takes values from the index set \textsf{Ind}, so the statistical distance between $X,Y$ is 

$$\Delta(X,Y)=\frac{1}{2}\sum_{\textsf{i}_i\in \textsf{Ind}}\big\lvert \Pr[X=\textsf{i}_i] - \Pr[Y=\textsf{i}_i]\big\rvert$$

As the statistical distance is computed from the concrete probability of each element, so we first evaluate the probability of the attack procedure correctly pointing out the real signer with the varying instance number for ring size $2$, $128$, $1024$, respectively (cf, Figure \ref{AttackOurRS:a}, \ref{AttackOurRS:b}, and \ref{AttackOurRS:c}). The experimental results  indicate that the evaluated probability ultimately converges to a uniform distribution. This is reaffirmed by the next experiment that we evaluate the distinguishing advantage based on these converged distributions, where the experimental result (cf. Figure \ref{AttackOurRS:d}) indicates the distinguishing advantage is finally converged to zero.

%用统计距离在这个分布（集合）上取值，用实验说明区分优势是可忽略的。

% \begin{figure}[!htb]
% \centering
% \ffigbox[\FBwidth]{
% \begin{subfloatrow}[2]
% \ffigbox[\FBwidth]{\caption{Clock time for \textsf{KeyGen}, \textsf{Sign}, and \textsf{Ver} with varying ring size}}{\includegraphics[width=5.1cm,height=3.8cm]{BenchmarkRSperformance.eps}}
% \ffigbox[\FBwidth]{\caption{Concrete signature size with varying ring size}}{\includegraphics[width=5.5cm,height=3.8cm]{BenchmarkRS_SigSize.eps}}
% \end{subfloatrow}
% \par\nointerlineskip\vspace{5mm}
% \begin{subfloatrow}[2]
% \ffigbox[\FBwidth]{\caption{Clock time for \textsf{KeyGen}, \textsf{Sign}, and \textsf{Ver} with varying ring size}}{\includegraphics[width=5.1cm,height=3.8cm]{BenchmarkRSperformance.eps}}
% \ffigbox[\FBwidth]{\caption{Concrete signature size with varying ring size}}{\includegraphics[width=5.5cm,height=3.8cm]{BenchmarkRS_SigSize.eps}}
% \end{subfloatrow}}
% {\caption{The concrete computation time and signature size of our ring signature scheme.}}
% %{\label{fig:BenchmarkPerformance}}
% \end{figure}

%Due to the Gaussian sampling is the most fundamental operation in this work. . Therefore, in order to guarantee the reliable of the experiment, we employ

\subsection{Unforgeability}\label{Sec:Unforgeability}\pdfbookmark[1]{Unforgeability}{Unforgeability}

We prove the unforgeability w.r.t. insider corruption of our ring signature scheme by giving a reduction from the hardness of the RingSIS$_{q,n,m,\beta}$ assumption (cf. the Definition \ref{DefofSIS}). The barrier of the reduction is how to simulate the corrupting and signing oracles when embedding the RSIS challenge in the ring. At the beginning of the simulation, the simulator embeds the RSIS challenge $\mathbf A\in R_q^{1\times m}$ in the ring by picking a random index $i^\diamond\in[N]$ then set $\mathbf A^{(i^\diamond)}=\mathbf A$. In this setting, the simulator can generate the trapdoor $\mathbf T^{(i)}$ by trapdoor generation algorithm for each one except the $i^\diamond$-th member, and hence the simulator can response the queries with respect to all ring members except the $i^\diamond$-th member. But under this setting, the question is how to simulate the signing oracle with respect to the index $i^\diamond$. We resolve that by the public trapdoor $\mathbf{T_G}$ of the gadget matrix $\mathbf{G}$, i.e., using $\mathbf{T_G}$ to extract signatures for the queries on index $i^\diamond$. But the public basis $\mathbf T_{\mathbf G}$ also can be exploited by the adversary to trivially forge signatures. To guard that, we add a step in the forgery output phase. Moreover, the Gaussian parameter $\sigma_{\textsc{Saml}}$ should be compensated so that the signatures distributed statistically indistinguishable between the real and simulated world. Additionally, we employ the following simulation tools, the algorithms of \textsf{SampleR} and \textsf{TrapExdABB}, both of which can directly transfer to ring setting since the underlying algorithm is actually the Gaussian sampling algorithm \textsf{GauSample}.

\begin{lemma}[{\rm \textsf{SampleR}} Algorithm {\rm\cite{ABB10b}}]
There is a PPT preimage sampling algorithm {\rm $\textsf{SampleR}(\mathbf B, \mathbf {T_B},\mathbf A,\sigma_{\textsc{tg}})$} which takes as input $\mathbf B\in R_q^{1\times m}$ and its a trapdoor $\mathbf {T_B}$, a matrix $\mathbf A\in R_q^{1\times m}$, and a standard deviation {\rm$\sigma_{\textsc{tg}}$}. Then it outputs a matrix $\mathbf R\in R^{m\times m}$ and its distribution is statistically close to {\rm$D_{R^{m\times m}, \sigma_{\textsc{tg}}}$} such that $\mathbf B=\mathbf A\mathbf R^{-1}\pmod q$ and {\rm$\lVert\mathbf R\rVert \leq \sigma_{\textsc{tg}} \sqrt{m}$}. 
\label{SampleR}
\end{lemma}
\vspace{-1mm}

The \textsf{SampleR} algorithm is actually a sub-algorithm of the \textsf{SampleRwithBasis} algorithm of~\cite{ABB10b}, which is used to build the relation between the RingSIS challenge $\mathbf A$ and the ring of verification keys.

\begin{lemma}[{\rm \textsf{TrapExdR}} Algorithm~{\rm\cite{ABB10a}}]
There is a PPT trapdoor extend algorithm {\rm $\textsf{TrapExdR}(\mathbf A,\mathbf {B},\mathbf R, \mathbf {T_B})$} which on input $\mathbf A\in R_q^{1\times m}$, $\mathbf B\in R_q^{1\times m}$ and its trapdoor $\mathbf {T_B}$, and $\mathbf R\in R^{m\times m}$, outputs an extended trapdoor $\mathbf {T_F}$ of $\mathbf F$ such that $s_1(\mathbf {T_F})<(\lVert {\mathbf {R}}\rVert+1)\cdot s_1(\mathbf {T_B})$ where $\mathbf F=[\mathbf A \mid \mathbf {AR}+\mathbf B]\in R_q^{1\times 2m}$.
\label{TrapExdR}
\end{lemma}
  \vspace{-1mm}

The \textsf{TrapExdR} algorithm is actually a sub-algorithm of the \textsf{SampleRight} algorithm of~\cite{ABB10a}, which is useful in the simulating of signing oracle. Specifically, as aforementioned, the simulator cannot simulate the signing queries with respect to the index $i^\diamond$ since there is no corresponding trapdoor, which can be resolved by using \textsf{TrapExdR} algorithm to generate trapdoor from the public trapdoor $\mathbf{T_G}$ of $\mathbf{G}$.

\begin{theorem}[Unforgeability]\label{Them:UnforgeabilityProof}
% When set the parametes {\rm$\sigma_{\textsc{Saml}}=O((\sqrt{\ell}\log m + \sqrt{m})\cdot \sqrt{Nk})\cdot \omega({\log n})$} and {\rm$\beta=O((\sqrt{\ell}\log m + \sqrt{m})^2 \sqrt{Nk})\cdot\omega(\log n\sqrt{\ln nw})$}, the ring signature scheme presented in Section \ref{Construction} is unforgeable w.r.t. insider corruption in the standard model.
When set the parametes {\rm\begin{small}
$$\sigma_{\textsc{Saml}}=O(\sqrt{Nk\ell}\log m + \sqrt{Nkm})\cdot \omega({\log n}),\quad \beta=O(m^{3/2})\cdot \sigma_{\textsc{td}}\sigma_{\textsc{Saml}}$$\end{small}}the ring signature scheme presented in Section \ref{Construction} is unforgeable w.r.t. insider corruption in the standard model.
\end{theorem}

\begin{proof}

Consider the following security game between an adversary $\mathcal{A}$ and the reduction $\mathcal{S}$. Upon receiving a challenge $\mathbf A\in R_q^{1\times m}$, $\mathcal{S}$ simulates as follows. 

\noindent$\textbf{Setup Phase}$. $\mathcal{S}$ invokes $\textsf{PP}\leftarrow \textsf{Setup}(1^n;\rho_{\textsc{st}})$ and does as follows.

\begin{itemize}
	\item Select a random index $i^\diamond\in [N]$, set $\mathbf{A}^{(i^\diamond)}=\mathbf{A}$. For the remaining index $i\neq i^\diamond$ i.e., for the $i\in [N]\setminus i^\diamond$, it does as follows:
	\begin{enumerate}
		\item Compute $(\mathbf A^{(i)}, \mathbf T^{(i)})\leftarrow\textsf{TrapGen}(q,n,\sigma_{\textsc{tg}}; \rho_\textsc{kg}^{(i)})$. 
		
		\item If $\mathbf A^{(i)}$ is lexicographically behind the $\mathbf A^{(i^\diamond)}$, else goto Step 1.
  
        \textsf{This step is to ensure $\mathbf A^{(i^\diamond)}$ is lexicographically first in the set $\{\mathbf A^{(i)}\}_{i\in[N]}$.}
		
		\item Sample $\bar{\mathbf R}^{(i)}\leftarrow \textsf{SampleR}(\mathbf A^{(i)}, \mathbf T^{(i)}, \mathbf A^{(i^\diamond)}, \sigma_{\textsc{tg}})$. By Lemma \ref{SampleR}, it holds that $\mathbf A^{(i)} = \mathbf A^{(i^\diamond)}{(\bar{\mathbf R}^{(i)})^{-1}}$. 
	\end{enumerate}

	\item For $i\in[N]$ and $d\in \{0,1\}$: Choose $\mathbf R_{\mathbf A_{d}^{(i)}}, \mathbf R_{\mathbf C_{d}^{(i)}}\xleftarrow{_\$} \{1,-1\}^{m\times m}$. Construct the ``PRF output'' matrix $\mathbf A_{d}^{(i)}=\mathbf A^{(i)}\mathbf R_{\mathbf A_{d}^{(i)}}+d\mathbf G$ and the ``message matching'' matrix $\mathbf C_{d}^{(i)}=\mathbf A^{(i)}\mathbf R_{\mathbf C_{d}^{(i)}}+d\mathbf G$. 
    
	\item For $i\in[N]$, select a PRF key $\mathbf k^{(i)}=(k_1^{(i)},\dots,k_k^{(i)})\xleftarrow{_\$}\{0,1\}^\kappa$.
	
    \item For $i\in[N]$, $j\in[\kappa]$, choose $\mathbf R_{\mathbf B_{j}^{(i)}}\xleftarrow{_\$} \{1,-1\}^{m\times m}$ and construct the ``PRF key matching'' matrix $\mathbf B_{j}^{(i)}=\mathbf A^{(i)}\mathbf R_{\mathbf B_{j}^{(i)}}+k_j^{(i)}\mathbf G$. 
    
    \item For $i\in[N]\setminus i^\diamond$, select $\mathbf u^{(i)}\xleftarrow{_\$}R_q$.
    \item For $i=i^\diamond$, sample a random preimage $\bar{\mathbf e}\xleftarrow{_\$}D_{R^{m},\sigma_{\textsc{tg}}}$, set $\mathbf u^{(i^\diamond)}=\mathbf A^{(i^\diamond)}\bar{\mathbf e}$.

    \item Set $\textsf{S}=\{\textsf{vk}^{(i)}\}_{i\in[N]}$ where each $\textsf{vk}^{(i)}=(\mathbf A^{(i)},(\mathbf A_{0}^{(i)},\mathbf A_{1}^{(i)}), \{\mathbf B_{j}^{(i)}\}_{j\in[\kappa]},(\mathbf C_{0}^{(i)},$ $\mathbf C_{1}^{(i)}), \mathbf u^{(i)})$, then sends (\textsf{PP}, $\textsf{S}, \rho_{\textsc{st}}$) to $\mathcal{A}$.
    \end{itemize}

\noindent$\textbf{Simulating Signing Oracle}$. For the query tuple $(\boldsymbol{\mu},\textsf{R},s)$ provided by the adversary, $\mathcal{S}$ replies that as below.

\begin{itemize}\setlength{\itemsep}{2pt} \setlength{\parsep}{0pt} \setlength{\parskip}{1pt}
		
	\item Compute ${d}=\textsf{PRF}({\mathbf k}^{(s)}, \boldsymbol{\mu})$.
	\item Select $\mathbf u^{(i')}\leftarrow\{\mathbf u^{(i)}\}_{i\in[N]}$ for which the corresponding $\mathbf A^{(i')}$ is lexicographical first in the set $\{\mathbf A^{(i)}\}_{i\in[N]}$.

	\item For $i\in[N]$, compute $\mathbf A_{\boldsymbol{\mu}}^{(i)}=\textsf{Eval}(C_\textsf{PRF},(\{\mathbf B_{j}\}_{j\in[\kappa]}^{(i)},\mathbf C_{\mu_1}^{(i)},\dots,\mathbf C_{\mu_t}^{(i)}))$, that is\[\mathbf A_{\boldsymbol{\mu}}^{(i)}=\mathbf A^{(i)} \mathbf R_{\boldsymbol{\mu}}^{(i)}+d\mathbf G\]Note that the PRF function with inputs (message $\boldsymbol{\mu}$ and PRF key ${\mathbf k}^{(s)}$) is also evaluated since the message $\boldsymbol{\mu}$ was embed in the matrix $\mathbf C_{d}^{(i)}$ while the PRF key was in matrix $\mathbf B_{j}^{(i)}$.
	
	\item Let $\mathbf F_{\boldsymbol{\mu}, 1-{d}}^{(i)}=\big[\mathbf A^{(i)} \mid \mathbf A_{1-{d}}^{(i)}-\mathbf A_{\boldsymbol{\mu}}^{(i)}\big]$, we have
	\begin{align*}
		\mathbf F_{\boldsymbol{\mu}, 1-{d}}^{(i)}=\Big[\mathbf A^{(i)}\mid \mathbf A^{(i)} \big(\mathbf R_{\mathbf A_{1-{d}}^{(i)}}-\mathbf R_{\boldsymbol{\mu}}^{(i)}\big)+(1-2d)\mathbf G\Big].
	\end{align*}

	\item If $i\neq i^\diamond$, use $\mathbf T^{(i)}$ to sample the preimage $\mathbf e$ as \textsf{Sign} algorithm. 
	
	\item If $i= i^\diamond$, it first extends the public trapdoor $\mathbf T_{\mathbf G}$ to the augmented trapdoor $\mathbf T_{\mathbf F_{\boldsymbol{\mu}, 1-d}^{(i^\diamond)}}$ by the \textsf{TrapExdR} algorithm \begin{align*}\mathbf T_{\mathbf F_{\boldsymbol{\mu}, 1-d}^{(i^\diamond)}}\leftarrow\textsf{TrapExdR}\big(\mathbf A^{(i^\diamond)}, \mathbf{G}, \big(\mathbf R_{\mathbf A_{1-{d}}^{(i)}}-\mathbf R_{\boldsymbol{\mu}}^{(i)}\big), \mathbf{T_G}\big)\end{align*}Then let $\mathbf F'_{\boldsymbol{\mu},1-d}=\big[\mathbf F_{\boldsymbol{\mu},1-d}^{(1)}\mid\dots\mid\mathbf F_{\boldsymbol{\mu},1-d}^{(N)}\big]$, the delegated trapdoor is  \begin{align*}\mathbf T_{\mathbf F'_{\boldsymbol{\mu},1-d}}\leftarrow\textsf{TrapDel}\Big(\mathbf T_{\mathbf F_{\boldsymbol{\mu}, 1-d}^{(i^\diamond)}}, \mathbf F'_{\boldsymbol{\mu},1-d}, \sigma_{\textsc{td}}\Big)\end{align*}Now we can sample $\mathbf e$ by $\mathbf e\leftarrow\textsf{GauSample}(\mathbf F'_{\boldsymbol{\mu},1-d}, \mathbf T_{\mathbf F'_{\boldsymbol{\mu},1-d}}, \mathbf u^{(i')}, \sigma_{\textsc{Saml}})$. 
	
	%As constructed in Setup phase, the simulator produced the secret key for each ring member except the ring member with index $i^\diamond$. So the simulator cannot response the signing query for index $i^\diamond$ as other indexes. By our careful parameters setting, the simulator can employ the public basis $\mathbf{T_G}$ of the gadget matrix $\mathbf{G}$ to compute the extended basis as the \textsf{Sign} procedure by \textsf{BasisExdABB}, \textsf{BasisExd}, and \textsf{RandBasis} in a similar manner. 
	
	\item Return the signature $\Sigma=\mathbf e$ for the tuple $(\boldsymbol{\mu},\textsf{R},s)$ to $\mathcal{A}$ and adds $(\boldsymbol{\mu},\textsf{R},\Sigma)$ to a list {\rm \textsf{L}} which $\mathcal{S}$ initialized in prior.
	
\end{itemize}

\noindent$\textbf{Simulating Corrupting Oracle}$. $\mathcal{A}$ queries the corrupting oracle $\textsf{OCorrupt}(\cdot)$ with index $i$, $\mathcal{S}$ returns the randomness $\rho_{\textsc{kg}}^{(i)}$ to $\mathcal{A}$ and adds {\rm $\textsf{vk}^{(i)}$} to a set {\rm \textsf{C}} which $\mathcal{S}$ initialized in prior, while if $i=i^\diamond$ then $\mathcal{S}$ aborts. 

\noindent $\textbf{Exploiting the Forgery}$. $\mathcal{A}$ outputs a forgery $(\boldsymbol{\mu}^*,\textsf{R}^*,\Sigma^*)$. Let $N^*=\lvert\textsf{R}^*\rvert$. Parse $\boldsymbol{\mu}^*=(\mu_1^*,\dots,\mu_t^*)$, $\textsf{R}^*=({\textsf{vk}}^{(1)},\dots,{\textsf{vk}}^{(N^*)})$, and $\Sigma^*=\mathbf e^*$. $\mathcal{S}$ does the following to exploit the forgery.

\begin{itemize}\setlength{\itemsep}{5pt} \setlength{\parsep}{0pt} \setlength{\parskip}{1pt}
	\item Check if $(\boldsymbol{\mu}^*,\textsf{R}^*,\Sigma^*) \in \textsf{L}$ or $\lVert {\mathbf e^{(i^*)}}\rVert> \sigma_{\textsc{Saml}}\sqrt{m}$ or $i^\diamond\notin \textsf{R}^*$ or exists $\textsf{vk}^{(i^*)}\in \textsf{C}$, $\mathcal{S}$ aborts. This step is used to guarantee that the forgery  $(\boldsymbol{\mu}^*,\textsf{R}^*,\Sigma^*)$ is well-formed, had not been queried to the signing oracle and corrupting oracle, and the desired index $i^\diamond$ is included in the ring $\textsf{R}^*$. Note that the event that the index $i^\diamond$, that the simulator previously selected in the Setup phase, is involved in $\textsf{R}^*$ happened less than $\frac{1}{N}$, so there is a reduction loss of $\frac{1}{N}$.
	
	\item Compute $d^*=\textsf{PRF}(\mathbf k^{(i^\diamond)},\boldsymbol{\mu}^*)$. Choose $\mathbf u^{(i')}$ from $\{\mathbf u^{(i^*)}\}_{i\in[N^*]}$ for which the corresponding $\mathbf A^{(i')}$ is lexicographically first in $\{\mathbf A^{(i^*)}\}_{i^*\in[N^*]}$. Due to our prior setting in \textbf{Setup} phase, it holds that $\mathbf u^{(i')}=\mathbf u^{(i^\diamond)}$.

	\item For $i^*\in[N^*]$, compute $\mathbf A_{\boldsymbol{\mu}^*}^{(i^*)}=\mathbf A^{(i^*)}\mathbf R_{\boldsymbol{\mu}^*}^{(i^*)}+\textsf{PRF}\big(\mathbf k^{(i^\diamond)},\boldsymbol{\mu}^*\big)\mathbf G$ by\\ \begin{center}$\mathbf A_{\boldsymbol{\mu}^*}^{(i^*)}=\textsf{Eval}\Big(C_\textsf{PRF},\big(\{\mathbf B_{j}\}_{j\in[k]}^{(i^*)},\mathbf C_{\mu_1^*}^{(i^*)},\mathbf C_{\mu_2^*}^{(i^*)},\dots,\mathbf C_{\mu_t^*}^{(i^*)}\big)\Big)$.\end{center}
	
	\item For $i^*\in[N^*]$, construct\\
	\begin{center}$\mathbf F_{\boldsymbol{\mu}^*, 1-d^*}^{(i^*)}=\big[\mathbf A^{(i^*)}\mid \mathbf A_{1-d^*}^{(i^*)}-\mathbf A_{\boldsymbol{\mu}^*}^{(i^*)}\big]$,\ $\mathbf F'_{\boldsymbol{\mu}^*,1-d^*}=\big[\mathbf F_{\boldsymbol{\mu}^*,1-d^*}^{(1)}\mid\dots\mid\mathbf F_{\boldsymbol{\mu}^*,1-d^*}^{(N^*)}\big]$.\end{center}\hspace{0.51mm}$\mathbf F_{\boldsymbol{\mu}^*, d^*}^{(i^*)}=\big[\mathbf A^{(i^*)}\mid\mathbf A_{d^*}^{(i^*)}-\mathbf A_{\boldsymbol{\mu}^*}^{(i^*)}\big]$,\quad \hspace{3.3mm}$\mathbf F'_{\boldsymbol{\mu}^*,d^*}=\big[\mathbf F_{\boldsymbol{\mu}^*,d^*}^{(1)}\mid\dots\mid\mathbf F_{\boldsymbol{\mu}^*,d^*}^{(N^*)}\big]$.

	\item Check if $\mathbf F'_{\boldsymbol{\mu}^*,1-d^*}\cdot \mathbf e^*=\mathbf u^{(i^\diamond)}\pmod q$ holds, $\mathcal{S}$ aborts. This step is used to check if the forgery is produced by employing the public basis $\mathbf{T_G}$. 
	\item Check if $\mathbf F'_{\boldsymbol{\mu}^*,d^*}\cdot \mathbf e^*=\mathbf u^{(i^\diamond)}\pmod q$ holds, otherwise $\mathcal{S}$ aborts. The step is used to verify if the ring equation is hold.
	\item Note that the equation $\mathbf F'_{\boldsymbol{\mu}^*,d^*}\cdot \mathbf e^*=\mathbf u^{(i^\diamond)}\pmod q$ can be transformed to the following
	\begin{small}\[
	\sum_{i^*\in[N^*]}\big[\mathbf A^{(i^*)}\mid\mathbf A_{d^*}^{(i^*)}-\mathbf A_{\boldsymbol{\mu}^*}^{(i^*)}\big]\cdot \mathbf e^{(i^*)}=\mathbf u^{(i^\diamond)}
	\]\end{small}Recall the ``PRF output matching'' matrix $\bold A_{d^*}^{(i)}=\bold A^{(i)}\bold R_{\bold A_{d^*}^{(i)}}+d^*\bold G$ constructed in Setup phase. It holds that
	\begin{small}\[
	\sum_{i^*\in[N^*]}\big[\mathbf A^{(i^*)}\mid\mathbf A^{(i^*)}\big(\mathbf R_{\mathbf A_{d^*}^{(i^*)}}-\mathbf R_{\boldsymbol{\mu}^*}^{(i^*)}\big)+(d^*-\textsf{PRF}(\mathbf k^{(i^\diamond)},\boldsymbol{\mu}^*))\mathbf{G}\big]\cdot \mathbf e^{(i^*)}=\mathbf u^{(i^\diamond)}
	\]\end{small}
	\begin{small}\[
	\sum_{i^*\in[N^*]}\big[\mathbf A^{(i^*)}\mid\mathbf A^{(i^*)}\big(\mathbf R_{\mathbf A_{d}^{(i^*)}}-\mathbf R_{\boldsymbol{\mu}^*}^{(i^*)}\big)\big]\cdot \mathbf e^{(i^*)}=\mathbf u^{(i^\diamond)}
	\]\end{small}
	
	\item To avoid the cumbersome notation, 
	
	let $\hat{\mathbf{R}}^{(i^*)}=\mathbf R_{\mathbf A_{d}^{(i^*)}}-\mathbf R_{\boldsymbol{\mu}^*}^{(i^*)}$ and $\hat{\mathbf{R}}^{(i^\diamond)}=\mathbf R_{\mathbf A_{d}^{(i^\diamond)}}-\mathbf R_{\boldsymbol{\mu}^*}^{(i^\diamond)}$, so we have
	\begin{small}\begin{equation*}
	\sum_{i^*\in[N^*]}\big[\mathbf A^{(i^*)}\mid\mathbf A^{(i^*)}\hat{\mathbf{R}}^{(i^*)}\big]\cdot \mathbf e^{(i^*)}={\mathbf A}^{(i^\diamond)}\bar{\mathbf e}
	\end{equation*}\end{small}
	\item By the construction on $\{\mathbf{A}^{i^*}\}_{i^*\in[N^*]\setminus i^\diamond}$ in the \textsf{Setup} phase, it holds that
	\begin{scriptsize}\begin{equation*}
	\big[\mathbf A^{(i^\diamond)}\mid\mathbf A^{(i^\diamond)}\hat{\mathbf{R}}^{(i^\diamond)}\big] \mathbf e^{(i^\diamond)}
	+
	\sum_{i^*\in[N^*]\setminus i^\diamond}\big[\mathbf A^{(i^\diamond)}{(\bar{\mathbf R}^{(i^*)})^{-1}}\mid\mathbf A^{(i^\diamond)}{(\bar{\mathbf R}^{(i^*)})^{-1}}\hat{\mathbf{R}}^{(i^*)}\big] \mathbf e^{(i^*)}=\mathbf u^{(i^\diamond)}
	\end{equation*}\end{scriptsize}
	
	\item Let $\mathbf e^{(i^*)}=(\mathbf e_0^{(i^*)}\mid\mathbf e_1^{(i^*)})$ and $\mathbf e^{(i^\diamond)}=(\mathbf e_0^{(i^\diamond)}\mid\mathbf e_1^{(i^\diamond)})$. Note that $\mathbf u^{(i^\diamond)}=\mathbf A^{(i^\diamond)}\bar{\mathbf e}$ that we set in Setup phase. Therefore, it holds that $\mathbf A^{(i^\diamond)} \hat{\mathbf e}= \mathbf 0 \pmod q$ where 
\begin{small}\[
	\hat{\mathbf e}=\bigg(\mathbf e_0^{(i^\diamond)}+\hat{\mathbf{R}}^{(i^\diamond)}\mathbf e_1^{(i^\diamond)}+\sum_{i^*\in[N^*]\setminus i^\diamond}(\bar{\mathbf R}^{(i^*)})^{-1}\mathbf e_0^{(i^*)} + (\bar{\mathbf R}^{(i^*)})^{-1}\hat{\mathbf{R}}^{(i^*)} \mathbf e_1^{(i^*)} \bigg)- \bar{\mathbf e}
	\]\end{small}

	\item Return $\hat{\mathbf e}$ as the RingSIS$_{q,n,m,\beta}$ solution.
	\end{itemize}
	
\end{proof}

\vspace{1mm}
\begin{claim}\label{Claim7}
The set of verification keys {\rm $\textsf{S}$} that simulated by {\rm $\mathcal{S}$} has the correct distribution.
\end{claim}

\vspace{-3mm}
\begin{proof}
In the real scheme, all the matrices $\{\mathbf A^{(i)}\}_{i\in[N]}$ of each ring member's $\textsf{vk}^{(i)}$ is generated by $\textsf{TrapGen}$. In the simulation, only the matrix with index $i^\diamond$, i.e., the matrix $\mathbf A^{(i^\diamond)}$ is chosen from uniform since it comes from the challenger, while the remained matrices, i.e., $\{\mathbf A^{(i)}\}_{i\in[N]\setminus i^\diamond}$ are also generated by $\textsf{TrapGen}$. Therefore, by Lemma \ref{TrapGen}, the property of $\textsf{TrapGen}$ is that the output matrix which is statistical close to the uniform distribution. For the vectors $\{\mathbf u^{(i)}\}_{i\in[N]\setminus i^\diamond}$, it is immediate that both distributions are statistically indistinguishable, since both of which were uniformly random selected. For the matrices $(\mathbf A_{0}^{(i)},\mathbf A_{1}^{(i)}), \{\mathbf B_{j}^{(i)}\}_{j\in[\kappa]}$, $(\mathbf C_{0}^{(i)},\mathbf C_{1}^{(i)})$, and the vector $\mathbf u^{(i^\diamond)}$, all of which are constructed with the random item $\mathbf A^{(i)}$ or $\mathbf R^{(i)}$, so we can conclude which are statistically close to uniform. Consequently, the distribution on the set of verification keys $\textsf{S}$ is statistically indistinguishable from those in the real attack. 
\end{proof}

\vspace{-1mm}
\begin{claim}\label{Claim8}
When set $\sigma_{\textsc{Saml}}=O(\sqrt{Nk\ell}\log m + \sqrt{Nkm})\cdot \omega({\log n})$, the replies from {\rm \textsf{OSign}$(\cdot,\cdot,\cdot)$} that simulated by $\mathcal{S}$ has the correct distribution.
\end{claim}
\vspace{-4mm}
\begin{proof}
Let $\sigma_{\textsc{Saml,r}}$ and $\sigma_{\textsc{Saml,s}}$ be the Gaussian parameter for sampling preimages in real scheme and simulation, respectively. Recall the fact given in Lemma \ref{GauSample}, when the Gaussian standard deviation $\sigma$ is set sufficiently large, the preimage $\mathbf{e}$ is distributed statistically indistinguishable with the uniform. Therefore, we should set the $\sigma_{\textsc{Saml}}$ as the larger one between $\sigma_{\textsc{Saml,r}}$ and $\sigma_{\textsc{Saml,s}}$, i.e., $\sigma_{\textsc{Saml}}=\max\{\sigma_{\textsc{Saml,r}}, \sigma_{\textsc{Saml,s}}\}$. The answer is obvious, that is $\sigma_{\textsc{Saml}}= \sigma_{\textsc{Saml,s}}$. This is because the input trapdoor $\mathbf T_{\mathbf F_{\boldsymbol{\mu}, 1-d}^{(i^\diamond)}}$ of the \textsf{TrapDel} algorithm has a larger norm than in the real scheme, since the $\mathbf T_{\mathbf F_{\boldsymbol{\mu}, 1-d}^{(i^\diamond)}}$ is derived from the evaluated matrix $\mathbf R_{\boldsymbol{\mu}}^{(i)}$ (cf. the phase of the simulating signing oracle above). Concretely, we analyse as follows. Recall in the $\textsf{Simulating Signing Oracle}$ phase, the signature $\mathbf{e}$ with respect to the index $i^\diamond$ was responsed as below. 
\[
\mathbf F_{\boldsymbol{\mu}, 1-d}^{(i^\diamond)}=\Big[\mathbf A^{(i^\diamond)}\mid \mathbf A^{(i^\diamond)} \big(\mathbf R_{\mathbf A_{1-d}^{(i^\diamond)}}-\mathbf R_{\boldsymbol{\mu}}^{(i^\diamond)}\big)+(1-2d)\mathbf G\Big]\in R_q^{1\times 2m}
\]The trapdoor $\mathbf T_{\mathbf F_{\boldsymbol{\mu}, 1-d}^{(i^\diamond)}}$ is generated by algorithm \textsf{TrapExdR}. Then the trapdoor $\mathbf T_{\mathbf F_{\boldsymbol{\mu}, 1-d}^{(i^\diamond)}}$ is taken as input in the \textsf{TrapDel} algorithm and outputs the delegated trapdoor $\mathbf T_{\mathbf F'_{\boldsymbol{\mu},1-d}}$. Finally, the signature $\mathbf{e}$ is sampled by the trapdoor $\mathbf T_{\mathbf F'_{\boldsymbol{\mu},1-d}}$. Below we will use the properties with respect to these algorithms to analyze the parameter $\sigma_{\textsc{Saml,s}}$.

\noindent For simplicity, we set $\bar{\mathbf R}^{(i)}=\mathbf R_{\mathbf A_{1-d}^{(i)}}-\mathbf R_{\boldsymbol{\mu}}^{(i)}$. 
\begin{itemize}\setlength{\itemsep}{2pt} \setlength{\parsep}{0pt} \setlength{\parskip}{1pt}
    \item By Lemma \ref{Lemma7}, $\lVert \bar{\mathbf R}^{(i)} \rVert \leq O(\ell\log m + \sqrt{m})$ for some constant $c$.
    
    \item By Lemma \ref{TrapExdR}, $s_1(\mathbf T_{\mathbf F_{\boldsymbol{\mu}, 1-d}^{(i^\diamond)}}) < \big(\lVert \bar{\mathbf R}^{(i)} +1\rVert\big) \cdot s_1(\mathbf {T_G})$. By our definition on the gadget matrix $\mathbf{G}$, the $s_1(\mathbf {T_G})$ is some constant in our setting.
    
    \item By Lemma \ref{TrapDel}, $s_1(\mathbf T_{\mathbf F'_{\boldsymbol{\mu},1-d}})\leq \sigma_{\textsc{td}} \cdot O(\sqrt{(N-1)k+w}+\sqrt{k}+\omega(\sqrt{\log n}))$ and $\sigma_{\textsc{td}}\geq s_1(\mathbf T_{\mathbf F_{\boldsymbol{\mu}, 1-d}^{(i^\diamond)}})\cdot \omega(\sqrt{\log n})$.

    \item By Lemma \ref{GauSample}, it requires to set $\sigma_{\textsc{Saml}}\geq s_1(\mathbf T_{\mathbf F'_{\boldsymbol{\mu},1-d}}) \cdot \omega(\sqrt{\log n})$.
\end{itemize}
     To satisfy these requirements, set $\sigma_{\textsc{Saml}}=O(\sqrt{Nk}\ell\log m + \sqrt{Nkm})\cdot \omega({\log n})$ is sufficient.
\end{proof}

\vspace{-1mm}
\begin{claim}\label{Claim13} When set $\beta=O(m^{3/2})\cdot \sigma_{\textsc{td}}\sigma_{\textsc{Saml}}$, $\mathcal{A}$ can generate a valid solution for {\rm RingSIS$_{q,n,m,\beta}$} with overwhelming probability.
\end{claim}
\vspace{-3mm}
\begin{proof}
We argue that 
\begin{small}\[
	\hat{\mathbf e}=\bigg(\mathbf e_0^{(i^\diamond)}+\hat{\mathbf{R}}^{(i^\diamond)}\mathbf e_1^{(i^\diamond)}+\sum_{i^*\in[N^*]\setminus i^\diamond}(\bar{\mathbf R}^{(i^*)})^{-1}\mathbf e_0^{(i^*)} + (\bar{\mathbf R}^{(i^*)})^{-1}\hat{\mathbf{R}}^{(i^*)} \mathbf e_1^{(i^*)} \bigg)- \bar{\mathbf e}
	\]\end{small}
that $\mathcal{S}$ finally outputs in simulation is a valid SIS$_{q,n,m,\beta}$ solution. 
\begin{itemize}
    \item We first explain it is sufficiently short i.e., has a low-norm bounded by the parameter $\beta$. Note that $\mathbf e_0^{(i^*)}$ and $\mathbf e_1^{(i^*)}$ follow the distribution $ D_{R^{m},\sigma_{\textsc{tg}}}$. 
    
    \begin{itemize}
        \item By Lemma \ref{GauSample}, $\lVert{\mathbf e_0^{(i^\diamond)}}\rVert,\lVert{\mathbf e_1^{(i^\diamond)}}\rVert$, $\lVert{\mathbf e_0^{(i^*)}}\rVert,\lVert{\mathbf e_1^{(i^*)}}\rVert$, and  $\lVert\bar{\mathbf e}\rVert\leq \sigma_{\textsc{Saml}}\sqrt{m}$.
        
        \item By Lemma \ref{SampleR}, $\lVert\bar{\mathbf R}^{(i^*)}\rVert\leq \sigma_{\textsc{tg}}\sqrt{m}$.
        
        \item Since $\hat{\mathbf{R}}^{(i^*)}, \hat{\mathbf{R}}^{(i^\diamond)}\in\{1,-1\}^{m\times m}\subset R_q^{(m/n) \times (m/n)}$, so its $\ell_2$-norm is bounded as $c\sqrt{m}$ for any constant $c$.
        
        %Recall $\hat{\mathbf{R}}^{(i^*)}=\mathbf R_{\mathbf A_{d}^{(i^*)}}-\mathbf R_{\boldsymbol{\mu}^*}^{(i^*)}$ and $\hat{\mathbf{R}}^{(i^\diamond)}=\mathbf R_{\mathbf A_{d}^{(i^\diamond)}}-\mathbf R_{\boldsymbol{\mu}^*}^{(i^\diamond)}$, the norm of both are dominated by the items $\mathbf R_{\boldsymbol{\mu}^*}^{(i^*)}$ and $\mathbf R_{\boldsymbol{\mu}^*}^{(i^\diamond)}$. 
        
        %\item By Lemma \ref{Lemma7}, $\big\lVert\mathbf R_{\boldsymbol{\mu^*}}^{(i^*)}\big\rVert$ and $\big\lVert\mathbf R_{\boldsymbol{\mu^*}}^{(i^\diamond)}\big\rVert\leq O(\sqrt{\ell}\log m + \sqrt{m})$.
    \end{itemize}
       To satisfy these requirements, it requires to set $$\beta = O(m^{3/2})\cdot \sigma_{\textsc{tg}}\sigma_{\textsc{Saml}}$$
       %$$\beta \geq O(\ell\log m + \sqrt{m})\cdot \sigma_{\textsc{tg}}\sigma_{\textsc{Saml}}=O((\sqrt{\ell}\log m + \sqrt{m})^2 \sqrt{Nk})\cdot\omega(\log n\sqrt{\ln nw})$$

    \item Then we argue that $\hat{\mathbf e}$ is a non-zero RSIS$_{q,n,m,\beta}$ solution with overwhelming probability. Note that $\bar{\mathbf{e}}$ is selected by the simulator, $\hat{\mathbf{R}}^{(i^\diamond)}$, $\bar{\mathbf R}^{(i^*)}$ and $\hat{\mathbf{R}}^{(i^*)}$ are statistically close to $D_{R^{m\times m}, \sigma_{\textsc{td}}}$ and remain hidden from the adversary.  Therefore, for a valid forgery $\mathbf{e}^*$, $\hat{\mathbf{e}}^*\neq \mathbf{0}$ holds with overwhelming probability.  
\end{itemize}

Finally, we analyze $\mathcal{A}$'s advantage. Assume the {PRF} is secure, $\mathcal{A}$ cannot distinguish PRF from random functions, it will randomly pick either $\{\mathbf A_0^{(i^*)}\}_{i^*\in[N^*]}$ or $\{\mathbf A_1^{(i^*)}\}_{i^*\in[N^*]}$ to make a forgery. Therefore, with $\frac 1 2$ chance $\mathcal{A}$ will forge the one that $\mathcal{S}$ will be able to use to break the SIS$_{q,n,m,\beta}$ problem. Moreover, the probability $\textsf{Pr}[i^\diamond \in \textsf{R}^*]\geq \frac 1 N$. Therefore, we have \[\textsf{Adv}_{\textsf{SIS}_{q,n,m,\beta}}(\mathcal{S}^\mathcal{A})\geq \textsf{Adv}_{\textsf{Unf}}(\mathcal{A})/{(2N)}-\textsf{Adv}_{\textsf{PRF}}-\mathrm{negl}(n)\] where $\mathrm{negl}(n)$ denote the negligible statistical error in the simulation. 
\end{proof}

\vspace{2mm}

\subsection{Anonymity}\label{Sec:Anonymity}\pdfbookmark[1]{Anonymity}{Anonymity}

We prove the unconditional anonymity of our ring signature scheme by two experiments $\textsf{E}_0$ and $\textsf{E}_1$. The $\textsf{E}_0$ and $\textsf{E}_1$ correspond to the anonymity experiment as defined in Section \ref{PrivacyModel} with $b = 0$ and  $b = 1$, respectively. We will show a reduction to prove that once if the experiments of $\textsf{E}_0$ and $\textsf{E}_1$ can be distinguished by any adversary (even unbounded) then the Trapdoor Indistinguishability property of the \textsf{TrapDel} algorithm (cf. Lemma \ref{TrapDel}) is not holds.

\begin{theorem}[Anonymity]\label{Them:AnonymityProof}
The ring signature scheme presented in Section \ref{Construction} is unconditional anonymity in the standard model.
\end{theorem}
\vspace{-3mm}
\begin{proof}
The proof proceeds in experiments $\textsf{E}_0$, $\textsf{E}_1$ such that $\textsf{E}_0$ (resp., $\textsf{E}_1$) corresponds to the experiment of \textsf{Anonymity} in Definition \ref{RS} with $b = 0$ (resp., $b = 1$). We will show a reduction to prove that once if the experiments of $\textsf{E}_0$ and $\textsf{E}_1$ can be distinguished by any adversary (even if unbounded) then the underlying hardness can be broken.

\noindent$\textsf{E}_0$: This experiment first generate the public parameters $\textsf{PP}\leftarrow\textsf{Setup}(1^n;\gamma_{\textsc{st}})$ and key pair $\{\textsf{vk}^{(i)},\textsf{sk}^{(i)}\}_{i\in[N]}$ by repeatedly invoking $\textsf{KeyGen}(\gamma_{\textsc{kg}}^{(i)})$, and $\mathcal{A}$ is given $(\textsf{PP}, \textsf{S}=\{\textsf{vk}^{(i)}\}_{i\in[N]})$ and the randomness $(\gamma_{\textsc{st}}, \{\gamma_{\textsc{kg}}^{(i)}\}_{i\in[N]})$. Then $\mathcal{A}$ outputs a message $\mu^*$ and two indexes $(s_0^*,s_1^*)$ such that $s_0^*\neq s_1^*$ and $s_0^*,s_1^*\in[N]$. Then for $l\in[Q]$, $\mathcal{A}$ provides a ring {\rm $\textsf{R}_l^*$} such that {\rm $\textsf{vk}^{(s_0^*)},\textsf{vk}^{(s_1^*)} \in \textsf{S}\cap \textsf{R}_l^*$}, the experiment computes $\Sigma_l^*\leftarrow\textsf{Sign}\big(\textsf{sk}^{(s_0^*)},\mu^*, \textsf{R}_l^*\big)$ and sends $\Sigma_l^*$ to $\mathcal{A}$. Finally, $\mathcal{A}$ outputs a bit $b'$.

\noindent $\textsf{E}_1$: This experiment is the same as experiment $\textsf{E}_0$ except that the experiment uses $\textsf{sk}^{(s_1^*)}$ to response the challenge rather than $\textsf{sk}^{(s_0^*)}$.

\noindent{Then we show that $\textsf{E}_0$ and $\textsf{E}_1$ are statistically indistinguishable for $\mathcal{A}$, which we do by giving a reduction from the trapdoor indistinguishability property of \textsf{TrapDel} algorithm given in Lemma \ref{TrapDel}.}

\noindent\textbf{Reduction.} Let $\mathcal{A}$ be an adversary that has the ability to distinguish the above described experiments $\textsf{E}_0$ and $\textsf{E}_1$ with a non-negligible advantage. Then we can construct an algorithm $\mathcal{S}$ which runs $\mathcal{A}$ as a subroutine for breaking the property of \textsf{TrapDel}. 

\noindent$\textbf{Simulating Setup Phase}$. $\mathcal{S}$ generates \textsf{PP} and $\textsf{S}=(\textsf{vk}^{(1)},\dots,\textsf{vk}^{(N)})$ as above experiments. $\mathcal{S}$ sends $(\textsf{PP}, \textsf{S}, \gamma_{\textsc{st}}, \{\gamma_{\textsc{kg}}^{(i)}\}_{i\in[N]})$ to $\mathcal{A}$.

\noindent$\textbf{Challenge}$. The challenge phase has two sub-phases:
\begin{itemize}
    \item $\mathcal{A}$ provides a message $\mu^*$ and two indexes $(s_0^*,s_1^*)$ such that $s_0^*\neq s_1^*$ and $s_0^*,s_1^*\in[N]$. $\mathcal{S}$ chooses $b\xleftarrow{_\$}\{0,1\}$.

    \item For $l\in[Q]$, $\mathcal{A}$ provides a ring {\rm $\textsf{R}_l^*$} such that {\rm $\textsf{vk}^{(s_0^*)},\textsf{vk}^{(s_1^*)} \in \textsf{S}\cap \textsf{R}_l^*$}, $\mathcal{S}$ response it as same as the \textsf{Sign} algorithm, except the delegation of trapdoor $\mathbf{T}_{\mathbf F'_{\boldsymbol{\mu},1-d}}^{(s_b^*)}$. Concretely, $\mathcal{S}$ sends the trapdoors $(\mathbf{T}^{(s_0^*)}, \mathbf{T}^{(s_1^*)})$ to its challenger, then the challenger pick a random bit $b$ and compute the $\mathbf{T}_{\mathbf F'_{\boldsymbol{\mu},1-d}}^{(s_b^*)}$ as response. Then $\mathcal{S}$ uses $\mathbf{T}_{\mathbf F'_{\boldsymbol{\mu},1-d}}^{(s_b^*)}$ to issue signature $\Sigma_{l}^*$ for the challenge. 

\end{itemize}

%$\mathcal{A}$ provides a challenge {\rm $(\mu^*, \{\textsf{R}_l^*\}_{l\in[Q]},s_0^*,s_1^*)$} such that $s_0^*\neq s_1^*$ and {\rm $\textsf{vk}^{(s_0^*)},\textsf{vk}^{(s_1^*)} \in \textsf{S}\cap \textsf{R}_1^*\cap \dots\cap\textsf{R}_Q^*$}.

\noindent $\textbf{Exploiting}$. When $\mathcal{A}$ outputs $b'$, $\mathcal{S}$ outputs $b'$.

Note that when the bit $b$ that challenger randomly selected is $b=0$ then the view of $\mathcal{A}$ is distributed exactly according to $\textsf{E}_0$, while if $b=1$ then the view of $\mathcal{A}$ is distributed exactly according to $\textsf{E}_1$. By Lemma \ref{TrapDel}, we know the delegated trapdoor $\mathbf{T}_{\mathbf F'_{\boldsymbol{\mu},1-d}}^{(s_b^*)}$ is statistically independent with the original trapdoors $(\mathbf{T}^{(s_0^*)}, \mathbf{T}^{(s_1^*)})$. By Lemma \ref{GauSample}, we know the signature $\Sigma^*$ is statistically independent with the delegated trapdoor $\mathbf{T}_{\mathbf F'_{\boldsymbol{\mu},1-d}}^{(s_b^*)}$. Therefore, by the trapdoor indistinguishability property of \textsf{TrapDel} (cf. Lemma \ref{Lem:TrapDelProperty}), $\textsf{E}_0$ and $\textsf{E}_1$ are statistically indistinguishable. 
\end{proof}

\section{Acknowledgements}
This research was supported by the National Natural
Science Foundation of China (62072305).

\bibliographystyle{splncs04}
\bibliography{sn-bibliography}

\begin{thebibliography}{10}
\providecommand{\url}[1]{\texttt{#1}}
\providecommand{\urlprefix}{URL }
\providecommand{\doi}[1]{https://doi.org/#1}

\bibitem{ABB10a}
Agrawal, S., Boneh, D., Boyen, X.: Efficient lattice {(H)IBE} in the standard
  model. In: {EUROCRYPT} 2010. LNCS, vol.~6110, pp. 553--572 (2010)

\bibitem{ABB10b}
Agrawal, S., Boneh, D., Boyen, X.: Lattice basis delegation in fixed dimension
  and shorter-ciphertext {HIBE}. In: {CRYPTO} 2010. LNCS, vol.~6223, pp.
  98--115 (2010)

\bibitem{BDH+19}
Backes, M., D{\"{o}}ttling, N., Hanzlik, L., et~al.: Ring signatures:
  Logarithmic-size, no setup - from standard assumptions. In: {EUROCRYPT} 2019.
  LNCS, vol. 11478, pp. 281--311 (2019)

\bibitem{BPR12}
Banerjee, A., Peikert, C., Rosen, A.: Pseudorandom functions and lattices. In:
  Advances in Cryptology--EUROCRYPT 2012: 31st Annual International Conference
  on the Theory and Applications of Cryptographic Techniques, Cambridge, UK,
  April 15-19, 2012. Proceedings 31. pp. 719--737. Springer (2012)

\bibitem{BLO18}
Baum, C., Lin, H., Oechsner, S.: Towards practical lattice-based one-time
  linkable ring signatures. In: {ICICS} 2018. LNCS, vol. 11149, pp. 303--322
  (2018)

\bibitem{BKM06}
Bender, A., Katz, J., Morselli, R.: Ring signatures: Stronger definitions, and
  constructions without random oracles. In: {TCC} 2006. LNCS, vol.~3876, pp.
  60--79 (2006)

\bibitem{BFL+18}
Bert, P., Fouque, P.A., Roux-Langlois, A., Sabt, M.: Practical implementation
  of ring-sis/lwe based signature and ibe. In: Post-Quantum Cryptography: 9th
  International Conference, PQCrypto 2018, Fort Lauderdale, FL, USA, April
  9-11, 2018, Proceedings 9. pp. 271--291. Springer (2018)

\bibitem{BDF+11}
Boneh, D., Dagdelen, {\"{O}}., Fischlin, M., et~al.: Random oracles in a
  quantum world. In: {ASIACRYPT} 2011. LNCS, vol.~7073, pp. 41--69 (2011)

\bibitem{BGG+14}
Boneh, D., Gentry, C., Gorbunov, S., et~al.: Fully key-homomorphic encryption,
  arithmetic circuit {ABE} and compact garbled circuits. In: {EUROCRYPT} 2014.
  LNCS, vol.~8441, pp. 533--556 (2014)

\bibitem{BDR15}
Bose, P., Das, D., Rangan, C.P.: Constant size ring signature without random
  oracle. In: {ACISP} 2015. LNCS, vol.~9144, pp. 230--247 (2015)

\bibitem{BK10}
Brakerski, Z., Kalai, Y.T.: A framework for efficient signatures, ring
  signatures and identity based encryption in the standard model. Cryptology
  ePrint Archive: Report 2010/086, 2010  (2010)

\bibitem{BV14}
Brakerski, Z., Vaikuntanathan, V.: Lattice-based {FHE} as secure as {PKE}. In:
  ITCS 2014. pp. 1--12 (2014)

\bibitem{BDW22}
Branco, P., D\"ottling, N., Wohnig, S.: Universal ring signatures in the
  standard model. {IACR} ePrint Arch. p.~1265 (2022),
  \url{https://eprint.iacr.org/2022/1265}

\bibitem{CGH04}
Canetti, R., Goldreich, O., Halevi, S.: The random oracle methodology,
  revisited. J. {ACM}  \textbf{51}(4),  557--594 (2004)

\bibitem{CHKP10}
Cash, D., Hofheinz, D., Kiltz, E., Peikert, C.: Bonsai trees, or how to
  delegate a lattice basis. In: {EUROCRYPT} 2010. LNCS, vol.~6110, pp. 523--552
  (2010)

\bibitem{CCL+22}
Chatterjee, R., Chung, K., Liang, X., et~al.: A note on the post-quantum
  security of (ring) signatures. In: {PKC} 2022. LNCS, vol. 13178, pp. 407--436
  (2022)

\bibitem{CGH+21}
Chatterjee, R., Garg, S., Hajiabadi, M., et~al.: Compact ring signatures from
  learning with errors. In: {CRYPTO} 2021. LNCS, vol. 12825, pp. 282--312
  (2021)

\bibitem{CWL+06}
Chow, S.S.M., Wei, V.K., Liu, J.K., et~al.: Ring signatures without random
  oracles. In: {ASIACCS} 2006. pp. 297--302 (2006)

\bibitem{DRS18}
Derler, D., Ramacher, S., Slamanig, D.: Post-quantum zero-knowledge proofs for
  accumulators with applications to ring signatures from symmetric-key
  primitives. In: PQCrypto 2018. LNCS, vol. 10786, pp. 419--440 (2018)

\bibitem{DKN+04}
Dodis, Y., Kiayias, A., Nicolosi, A., et~al.: Anonymous identification in ad
  hoc groups. In: {EUROCRYPT} 2004. LNCS, vol.~3027, pp. 609--626 (2004)

\bibitem{DOP05}
Dodis, Y., Oliveira, R., Pietrzak, K.: On the generic insecurity of the full
  domain hash. In: {CRYPTO} 2005. LNCS, vol.~3621, pp. 449--466 (2005)

\bibitem{DM14}
Ducas, L., Micciancio, D.: Improved short lattice signatures in the standard
  model. In: Advances in Cryptology--CRYPTO 2014. pp. 335--352. Springer (2014)

\bibitem{ES20}
Eaton, E., Song, F.: A note on the instantiability of the quantum random
  oracle. In: PQCrypto 2020. LNCS, vol. 12100, pp. 503--523 (2020)

\bibitem{EZS+19}
Esgin, M.F., Zhao, R.K., Steinfeld, R., et~al.: Matrict: Efficient, scalable
  and post-quantum blockchain confidential transactions protocol. In: {CCS}
  2019. pp. 567--584 (2019)

\bibitem{GPV08}
Gentry, C., Peikert, C., Vaikuntanathan, V.: Trapdoors for hard lattices and
  new cryptographic constructions. In: STOC 2008. pp. 197--206 (2008)

\bibitem{KKW18}
Katz, J., Kolesnikov, V., Wang, X.: Improved non-interactive zero knowledge
  with applications to post-quantum signatures. In: {CCS} 2018. pp. 525--537
  (2018)

\bibitem{LAZ19}
Lu, X., Au, M.H., Zhang, Z.: Raptor: {A} practical lattice-based (linkable)
  ring signature. In: {ACNS} 2019. LNCS, vol. 11464, pp. 110--130 (2019)

\bibitem{Lyu09}
Lyubashevsky, V.: Fiat-shamir with aborts: Applications to lattice and
  factoring-based signatures. In: Advances in Cryptology--ASIACRYPT 2009: 15th
  International Conference on the Theory and Application of Cryptology and
  Information Security, Tokyo, Japan, December 6-10, 2009. Proceedings 15. pp.
  598--616. Springer (2009)

\bibitem{Lyu12}
Lyubashevsky, V.: Lattice signatures without trapdoors. In: {EUROCRYPT} 2012.
  LNCS, vol.~7237, pp. 738--755 (2012)

\bibitem{LPR10}
Lyubashevsky, V., Peikert, C., Regev, O.: On ideal lattices and learning with
  errors over rings. In: Advances in Cryptology--EUROCRYPT. vol.~2010 (2010)

\bibitem{MS17}
Malavolta, G., Schr{\"{o}}der, D.: Efficient ring signatures in the standard
  model. In: {ASIACRYPT} 2017. LNCS, vol. 10625, pp. 128--157 (2017)

\bibitem{MBB+13}
Melchor, C.A., Bettaieb, S., Boyen, X., et~al.: Adapting lyubashevsky's
  signature schemes to the ring signature setting. In: {AFRICACRYPT} 2013.
  LNCS, vol.~7918, pp. 1--25 (2013)

\bibitem{MN17}
Mennink, B., Neves, S.: Optimal prfs from blockcipher designs. IACR
  Transactions on Symmetric Cryptology pp. 228--252 (2017)

\bibitem{MP12}
Micciancio, D., Peikert, C.: Trapdoors for lattices: Simpler, tighter, faster,
  smaller. In: {EUROCRYPT} 2012. LNCS, vol.~7237, pp. 700--718 (2012)

\bibitem{PS19}
Park, S., Sealfon, A.: It wasn't me! - repudiability and claimability of ring
  signatures. In: {CRYPTO} 2019. LNCS, vol. 11694, pp. 159--190 (2019)

\bibitem{RST01}
Rivest, R.L., Shamir, A., Tauman, Y.: How to leak a secret. In: {ASIACRYPT}
  2001. LNCS, vol.~2248, pp. 552--565 (2001)

\bibitem{NISTSHA3}
SHA, N.: Standard: Permutation-based hash and extendable-output functions
  (draft fips pub 202) (2014)

\bibitem{TKS+19}
Torres, W.A.A., Kuchta, V., Steinfeld, R., et~al.: Lattice ringct {V2.0} with
  multiple input and multiple output wallets. In: {ACISP} 2019. LNCS, vol.
  11547, pp. 156--175 (2019)

\bibitem{TSS+18}
Torres, W.A.A., Steinfeld, R., Sakzad, A., et~al.: Post-quantum one-time
  linkable ring signature and application to ring confidential transactions in
  blockchain (lattice ringct v1.0). In: {ACISP} 2018. LNCS, vol. 10946, pp.
  558--576 (2018)

\bibitem{TSS+20}
Torres, W.A.A., Steinfeld, R., Sakzad, A., et~al.: Post-quantum linkable ring
  signature enabling distributed authorised ring confidential transactions in
  blockchain. {IACR} Cryptol. ePrint Arch. p.~1121 (2020),
  \url{https://eprint.iacr.org/2020/1121}

\bibitem{WZZ18}
Wang, S., Zhao, R., Zhang, Y.: Lattice-based ring signature scheme under the
  random oracle model. Int. J. High Perform. Comput. Netw.  \textbf{11}(4),
  332--341 (2018)

\bibitem{ZHX+16}
Zhang, Y., Hu, Y., Xie, J., et~al.: Efficient ring signature schemes over
  {NTRU} lattices. Secur. Commun. Networks  \textbf{9}(18),  5252--5261 (2016)

\end{thebibliography}

\end{document}